\documentclass[journal]{IEEEtran}
\usepackage{cite}
\ifCLASSINFOpdf
  \usepackage[pdftex]{graphicx}
\else
\fi
\usepackage{amsmath,amssymb}

\usepackage{bbm}

\usepackage{algorithmic}
\ifCLASSOPTIONcompsoc
  \usepackage[caption=false,font=normalsize,labelfont=sf,textfont=sf]{subfig}
\else
  \usepackage[caption=false,font=footnotesize]{subfig}
\fi
\usepackage{enumitem}
\usepackage[dvipsnames]{xcolor}
\newcommand{\bzero}{\mathbf{0}}
\newcommand{\ba}{\mathbf{a}}
\newcommand{\bd}{\mathbf{d}}
\newcommand{\be}{\mathbf{e}}

\newcommand{\bs}{\mathbf{s}}
\newcommand{\bv}{\mathbf{v}}
\newcommand{\brho}{\boldsymbol{\rho}}

\newcommand{\bA}{\mathbf{A}}

\newcommand{\bX}{\mathbf{X}}
\newcommand{\bY}{\mathbf{Y}}
\newcommand{\bF}{\mathbf{F}}

\newcommand{\bW}{\mathbf{W}}
\newcommand{\bI}{\mathbf{I}}
\newcommand{\bE}{\mathbf{E}}
\newcommand{\bbC}{\mathbb{C}}
\newcommand{\bbD}{\mathbb{D}}
\newcommand{\bbR}{\mathbb{R}}
\newcommand{\bbS}{\mathbb{S}}
\newcommand{\bbT}{\mathbb{T}}
\newcommand{\bbY}{\mathbb{Y}}
\newcommand{\bbU}{\mathbb{U}}
\newcommand{\bbW}{\mathbb{W}}
\newcommand{\bbV}{\mathbb{V}}
\newcommand{\bbP}{\mathbb{P}}
\newcommand{\bbN}{\mathbb{N}}
\newcommand{\bbZ}{\mathbb{Z}}
\newcommand{\rme}{\mathrm{e}}

\newcommand{\rmdist}{\mathrm{dist}}
\newcommand{\rmRe}{\mathrm{Re}}

\newcommand{\bx}{\mathbf{x}}

\newcommand{\by}{\mathbf{y}}
\newcommand{\bz}{\mathbf{z}}
\newcommand{\bb}{\mathbf{b}}
\newcommand{\bt}{\mathbf{t}}
\newcommand{\bu}{\mathbf{u}}

\newcommand{\bH}{\mathbf{H}}

\newcommand{\bP}{\mathbf{P}}

\newcommand{\bU}{\mathbf{U}}
\newcommand{\bZ}{\mathbf{Z}}
\newcommand{\calB}{\mathcal{B}}
\newcommand{\calH}{\mathcal{H}}
\newcommand{\calG}{\mathcal{G}}
\newcommand{\calF}{\mathcal{F}}
\newcommand{\calJ}{\mathcal{J}}
\newcommand{\calK}{\mathcal{K}}

\newcommand{\calR}{\mathcal{R}}
\newcommand{\calQ}{\mathcal{Q}}

\newcommand{\comment}[1]{}

\newtheorem{theorem}{Theorem}
\newtheorem{lemma}{Lemma}

\newtheorem{condition}{Condition}

\allowdisplaybreaks

\begin{document}

\title{Unrolled Wirtinger Flow with Deep Decoding Priors for Phaseless Imaging}

\author{Samia~Kazemi,
        Bariscan~Yonel,~\IEEEmembership{Member,~IEEE,}
        and~Birsen~Yazici,~\IEEEmembership{Fellow,~IEEE}
\thanks{Manuscript submitted on August 10, 2021. Revised manuscript submitted on March 6, 2022. This work was supported in part by the Air Force Office of Scientific Research (AFOSR) under the agreement FA9550-19-1-0284, in part by Office of Naval Research (ONR) under the agreement N00014-18-1-2068, in part by the National Science Foundation (NSF) under Grant No ECCS-1809234 and in part by the United States Naval Research Laboratory (NRL) under the agreement N00173-21-1-G007. (\emph{Corresponding author: Birsen Yazici}.)}
\thanks{The authors are with the Department of Electrical, Computer and Systems Engineering, Rensselaer Polytechnic Institute, Troy, NY 12180 USA (e-mail: kazems@rpi.edu; yonelb2@rpi.edu; yazici@ecse.rpi.edu)}}



\maketitle

\begin{abstract}
We introduce a deep learning (DL) based network and an associated exact recovery theory for imaging from intensity-only measurements.
The network architecture uses a recurrent structure that unrolls the Wirtinger Flow (WF) algorithm with a deep decoding prior that enables performing the algorithm updates in a lower dimensional encoded image space.
We use a separate deep network (DN), referred to as the encoding network, for transforming the spectral initialization used in the WF algorithm to an appropriate initial value for the encoded domain.
The unrolling scheme models a fixed number of iterations of the underlying optimization algorithm into a recurrent neural network (RNN).
Furthermore, it facilitates simultaneous learning of the parameters of the decoding and encoding networks and the RNN.
We establish a sufficient condition to guarantee exact recovery under deterministic forward models.
Additionally, we demonstrate the relation between the Lipschitz constants of the trained decoding prior and encoding networks to the convergence rate of the WF algorithm.
We show the practical applicability of our method in synthetic aperture imaging using high fidelity simulation data from the PCSWAT software.
Our numerical study shows that the decoding prior and the encoding network facilitate improvements in sample complexity.
\end{abstract}


%
\IEEEpeerreviewmaketitle

\section{Introduction}
\label{sec:introduction}
\subsection{Motivation and Prior Art}
\IEEEPARstart{P}{haseless} imaging refers to the task of reconstructing an image from measurements whose magnitude or intensity values are available while the phase information is either missing or unreliable.
This challenging problem necessitates compensation either through hand-crafted prior information~\cite{soltanolkotabi2019structured} or significant measurement redundancy~\cite{candes2015phase_IEEE, jaganathan2015phase}.
In practical imaging applications with deterministic forward maps, hand-crafted priors may not be sufficiently descriptive of the underlying image domain to reduce the requirement of a large number of measurements~\cite{cai2016optimalThWF, wu2020hadamardWF}.
In this paper, we introduce a deep learning (DL) based phaseless imaging method that incorporates data-driven prior information for deterministic imaging problems with theoretical convergence and an exact recovery guarantee.

We consider the state-of-the-art phase retrieval methods that fall under two general categories: \emph{Wirtinger Flow} (WF) type algorithms~\cite{candes2015phase_IEEE, chen2015WF, yuan2019SWF, yonel2020deterministic} and \emph{DL-based approaches}~\cite{hand2018phase, shamshad2018robust, hand2020compressive, shamshad2020robust, Jagtap_2020_DD}.
The first category includes WF~\cite{candes2015phase_IEEE} and its variants~\cite{chen2015WF, zhang2017RWF, cai2016optimalThWF, yuan2019SWF, wu2020hadamardWF} which offer exact recovery guarantees based on non-convex optimization.
Unlike the earlier lifting-based convex phase-retrieval algorithms~\cite{candes2013phaselift, candes2015phase}, WF performs iterations in the signal space relieving the extensive computation and memory requirements.
However, classical WF requires an appropriate choice of the initial estimate, learning rate and high sample complexity of $\mathcal{O}(N \log N)$ under the Gaussian measurement model.
Several initial estimates for WF have been studied including the spectral estimation~\cite{candes2015phase_IEEE}, spectral estimation with sample truncation~\cite{zhang2018median} and more general sample processing functions
\cite{lu2020phase, luo2019optimal, yonel2020spectral},
linear spectral estimation~\cite{ghods2018linear}, orthogonality-promoting initialization~\cite{wang2018TAF} etc.
Original WF algorithm has been extended to include prior information~\cite{cai2016optimalThWF, yuan2019SWF} to reduce its sample complexity, most prominent of which is sparsity.
However, finding a hand-crafted optimal basis over which the unknown image is sparse can be challenging.
Other variants of WF aiming to reduce sample complexity include~\cite{chen2015WF, zhang2017RWF}.
However, the exact recovery theory of WF~\cite{candes2015phase_IEEE} and its variants
\cite{chen2015WF, zhang2017RWF, cai2016optimalThWF, yuan2019SWF}
relies on the assumption that the forward map is Gaussian.
This poses a fundamental limitation for imaging applications since the forward models are almost always deterministic.

Recently, in~\cite{yonel2020deterministic}, we introduced a mathematical framework for establishing an exact recovery guarantee for the WF algorithm involving deterministic forward maps under a sufficient condition that sets a concentration bound on the spectral matrix~\cite{candes2015phase_IEEE}.
This paves the way for the adoption of WF-type algorithms in a wide range of practical applications with provable performance guarantees.
However, this framework does not account for prior information about the image domain or study how the sufficient condition will be affected by the incorporation of such information.

The second category of state-of-the-art methods for phaseless imaging are practically attractive as they present a trade-off between the number of measurements and the training data, by solving the imaging problem in a lower dimensional encoded image space using a generative prior~\cite{hand2018phase, shamshad2018robust, hand2020compressive, shamshad2020robust}.
These are iterative algorithms where the parameters of the prior network, often referred to as the generative network, are learned to capture the global characteristics of the image manifold.
Once trained, starting from a randomly initialized encoded image, this network is used to update the encoded image estimation.
A convergence guarantee for the phaseless imaging problem is established for real positive-valued unknown image components in~\cite{hand2018phase} given that the trained weight matrices and the forward map satisfy a weight distribution condition and a range restricted concentration property, respectively.
For the same network as in~\cite{hand2018phase} with expansive layers of particular dimensionalities, and a measurement matrix and trained weight matrices of i.i.d. Gaussian distributed components,~\cite{hand2020compressive} shows that optimal sample complexity can be achieved for the phaseless imaging problem after a sufficient number of iterations.
However, since the prior network is trained separately from the phaseless imaging problem, these methods require large training sets in order to effectively estimate the probability distribution over the image domain instead of a conditional distribution given the phaseless measurements~\cite{willett2020_neuman}.
Additionally, this training scheme precludes the inclusion of an optimal initialization scheme for the encoded image space.

For overcoming the large training set requirement and fixed image space restriction of the generative prior, a related class of methods utilizes untrained networks in which the network structure itself works as the prior~\cite{Ulyanov_2018_DIP}.
For the phaseless imaging problem, a deep decoder~\cite{Heckle_2018_DD, Jagtap_2020_DD}, which uses an under-parameterized architecture, is utilized in~\cite{Jagtap_2020_DD} and an exact recovery guarantee is established for a two-layer decoder model and Gaussian distributed forward map that satisfies a specific restricted eigenvalue condition.
However, an optimal initialization scheme for the weights of the network, instead of the encoded image, is not established.
Additionally, theoretical results for this approach are very limited.

To address the limitations of state-of-the-art phaseless imaging methods, in this paper, we combine the WF algorithm and theory in~\cite{yonel2020deterministic} with a DL-based approach.
We consider the following two major modifications: the use of a deep decoding prior in conjunction with DL-based initialization and the unrolling of the WF algorithm into a recurrent neural network (RNN) architecture
which enables end-to-end training.
Our overall network is composed of the transformation network for initialization referred to as the \emph{encoder}, an RNN that represents the unrolled gradient descent updates of the WF in the encoded domain and the deep decoding prior network referred to as the \emph{decoder}.

Unrolling, which has been widely implemented to a range of linear inversion problems~\cite{willett2020_neuman, monga2021_unrolling} has limited utilization in the phase retrieval literature.
In~\cite{hyder2020solving}, an unrolled network is introduced for a Fourier phase retrieval problem with a reference signal.
In~\cite{zhang2021physics}, a complex unrolled network with unsupervised training is proposed for lensless microscopy imaging from phaseless measurements.
An unrolled Incremental Reshaped Wirtinger Flow based phase retrieval approach is presented in~\cite{naimipour2020_upr} for direct image estimation from amplitude measurements.
However, the trainable parameter set for this method is only related to the learning rates and no theoretical exact recovery guarantee is established.
To the best of our knowledge, our approach is the first to unroll a phaseless imaging algorithm with deep priors and end-to-end supervised training for general imaging applications.
Additionally, we have established a theoretical exact recovery guarantee.
A related approach in~\cite{Metzler2018_prdeep} incorporates adaptive step sizes, but their implementation does not use a fixed number of iterations, the step sizes are not learned and no theoretical exact recovery guarantee is established.

\subsection{Our Approach and its Advantages}
Our approach bridges the class of theoretically sound state-of-the-art purely optimization-based non-convex approaches with data-driven schemes deploying deep decoding priors for phaseless imaging in a deterministic setting.
Instead of the generative adversarial network (GAN)~\cite{radford2015unsupervised} based training used in the prior work~\cite{hand2018phase, shamshad2018robust, hand2020compressive, shamshad2020robust}, we adopt an end-to-end training approach where the parameters of the decoder, RNN and the encoder are learned simultaneously during training.
The unrolling strategy benefits from the inherent computational efficiency of a trained optimal network.
Additionally, being derived from model-based iterative algorithms, the network also offers interpretability of its architecture and parameters unlike an arbitrary deep network for phaseless imaging.

Our approach relates the spectral initialization-based WF algorithm with a generative prior based approach within a DL framework.
Existing applications of the generative prior~\cite{hand2018phase, shamshad2018robust, hand2020compressive, shamshad2020robust} lack a rigorous justification for the choice of initialization.
Furthermore, it is not well-understood how this value affects the convergence rate.
By establishing an explicit connection to the spectral initialization step, we determine the effect of the decoding network on the validity of the convergence guarantees and the rate of convergence to the true solution.
Our theoretical analysis reveals two key observations:
\begin{itemize}
    \item Firstly, the parameters of the underlying encoding and decoding prior networks have direct implications on the convergence rate and initialization accuracy which can be quantified by their Lipschitz constant values after training.
    A learned decoding prior can achieve a faster convergence rate compared to non-DL based WF~\cite{yonel2020deterministic} as long as certain Lipschitz constant related
    conditions are satisfied by the trained networks.
    \item Secondly, using the lower dimensional embedding of the decoding prior, we establish a new sufficient condition for exact recovery where, by virtue of specific imposed conditions on the decoder,
    the concentration property considered in~\cite{yonel2020deterministic} is parameterized over the encoded space.
    Hence, a sufficiently accurate initial estimate for the algorithm can be obtained using fewer measurements, as the representations are embedded in the lower dimensional space by the encoder.
    This sample complexity reduction aspect is also observed empirically through our numerical simulations.
\end{itemize}

The main differences with the existing generative prior based phase retrieval methods are notably in the initialization criteria, and the type of conditions assumed on the measurement vectors and the DL network parameters for establishing exact recovery guarantee when compared to~\cite{hand2018phase, hand2020compressive}.
In~\cite{hand2018phase, shamshad2018robust, hand2020compressive, shamshad2020robust}, the encoded unknown is randomly initialized, while in our approach, which can be viewed as a DL enhanced WF, we implement a DL network to transform the spectral initialization output to an encoded initialization value in order to facilitate a better starting point.
Even though the spectral initialization is computationally more expensive compared to a random initialization step, imaging applications in~\cite{shamshad2018robust, shamshad2020robust} use multiple initial guesses each of which is iteratively updated for selecting the best one.
Our approach avoids the need for repeating the algorithm for an arbitrary number of initial guesses, and its computation complexity is of the same order as in~\cite{yonel2019generalization}.
Additionally, unlike~\cite{hand2018phase, hand2020compressive}, our sufficient conditions on the trained DL networks for achieving exact recovery guarantee do not depend on the explicit consideration of the network architectures or imposition of specific properties on the trained network weights.
The sufficient condition on the forward map is similar to the deterministic WF analysis in~\cite{yonel2020deterministic} rather than the generative network architecture dependent condition in~\cite{hand2018phase, hand2020compressive}.

Our numerical simulation results demonstrate the ability of end-to-end learning with the unrolled WF method for reconstructing a wide range of unknown image sets.
This includes MNIST image set of handwritten digits, simulated images with geometric objects and PCSWAT~\cite{sammelmann2002personal} simulated images with mine-like objects for different non-Gaussian deterministic forward maps.

\vspace{-0.1in}
\subsection{Notation and Organization of the Paper}
\label{subsec:notation_organization}
Bold upper case and bold lower case letters are used to represent matrices and vectors, respectively.
$\|{\bX}\|_F$ refers to the Frobenius norm of $\bX$, and it is calculated as $\mathrm{Tr}(\bX^H\bX)$. $\mathrm{Tr}(.)$ denotes the trace of a matrix, while superscripts $T$ and $H$ on a matrix (or vector) denote its transpose and Hermitian transpose, respectively.
$\|.\|$ around a matrix and a vector refer to their spectral norm and $\ell_2$-norm, respectively.
Calligraphic letters and doublestruck upper case letters are used for operators and sets, respectively.
We use lower case Greek letters to represent various constants, and lower case italic letters, with or without subscripts, are used to denote different functions.
For a network $\calB$ with input $\bx$, $\calB(\bx)$ is its output vector.
Finally, we are using upper case italic letters for constant integers, and a set of integer values from $1$ to $K$ is written as $[K]$.
\begin{table}[!t]
\renewcommand{\arraystretch}{1.0}
\caption{List of Important Notations}
\label{table:notations}
\centering
\begin{tabular}{|c|c|}
\hline
$\hat{\brho}$ & Estimated image \\
\hline
$\brho^*$ & True unknown image \\
\hline
$\brho\brho^H$ & Lifted image vector $\brho$ \\
\hline
$\bd$ & Vector of the measured intensity values $\{d_m\}_{m = 1}^M$ \\
\hline
$\calF$ & Lifted forward map where $\bd = \calF(\brho^*{\brho^*}^H)$ \\
\hline
$\bF$ & Forward map with $\{\ba^H_m\}_{m = 1}^M$  along its rows \\
\hline
$\calF^H$ & $\ell_2$ adjoint of $\calF$ \\
\hline
$\bY$ & Spectral matrix defined as $\bY := \frac{1}{M}\calF^H(\bd)$ \\
\hline
$\calG$ & Encoding network or Encoder \\
\hline
$\calH$ & \shortstack{Decoding prior network or Decoding network \\
or Decoder} \\
\hline
$\bE_{\brho}$ & $\calH(\by)\calH(\by)^H - \brho^*{\brho^*}^H$ \\
\hline
$\langle ., .\rangle_F$ & Frobenius inner product \\
\hline
\end{tabular}
\end{table}
Table~\ref{table:notations} includes a list of important notations used throughout this paper.

The rest of the paper is organized as follows:
The problem statement and background on the non-DL based phase retrieval methods are discussed in Section~\ref{sec:prob_statement}.
The DL-based overall imaging network is introduced in Section~\ref{sec:DN}.
Theoretical foundations required for establishing the exact recovery guarantee of our approach are discussed in Section~\ref{sec:convergence}.
Section~\ref{sec:recovery_guarantees} presents our theoretical results involving the accuracy of the DL-based initial value, convergence guarantee and properties on the DNs for desired reconstruction performance.
The training process and the implementation details of specific properties of the encoder, decoder and the RNN are presented in Subsection~\ref{subsec:implement_Lipschitz_bound} and Subsection~\ref{subsec:complexity} discusses the computational complexity of our approach.
Numerical simulations examining the performance of our approach compared to the WF algorithm and other DL-based methods are presented in Section~\ref{sec:numerical_results}.
Finally, Section~\ref{sec:conclusions} concludes the paper.

\section{Problem Statement}
\label{sec:prob_statement}
\subsection{The Phase Retrieval Problem}
The phase retrieval problem entails estimating an unknown $\brho^* \in \mathbb{C}^N$, from its intensity, or magnitude-only measurements of the form:
\begin{equation}
\label{eq:phaless} d_m = |\langle\ba_m, \brho^*\rangle|^2, \quad \text{for} \ \ m = 1, 2, \ldots M,
\end{equation}
where $\ba_m\in\mathbb{C}^N$, for all $m = 1, \cdots, M$, denotes the $m^{th}$ sampling vector.
These vectors constitute a known, \emph{linear} measurement model, $\bF$, pertaining to the application of interest, such as Gaussian sampling, coded diffraction patterns, Fourier transform etc.
We refer to $\bF$ as the forward map.
When $\{\ba_m \}_{m = 1}^M$ are Fourier sampling vectors, the problem is classically known as Fourier phase retrieval, or \emph{the phase problem} in optical imaging, and quantum physics fields. 

Fundamentally, \eqref{eq:phaless} constitutes a system of $M$ quadratic equations, and solving it is known to be NP-hard in general \cite{wang2017solving}. 
Nonetheless, classical algorithms based on alternating minimization have been used to empirical success in optical imaging applications
\cite{fienup1982phase, netrapalli2013phase, luke2004relaxed},
despite the severe ill-posedness of the problem that arises due to the quadratic dependence of the measurements to the quantity of interest in \eqref{eq:phaless}
\cite{barnett2020geometry}. 

Over the last decade, optimization-based approaches have methodically progressed towards establishing performance guarantees in exactly recovering $\brho^*$ from $\bd = [d_1, \cdots d_M]^T \in \mathbb{R}^M$.
First major developments to this end have been through a reformulation of \eqref{eq:phaless} via \emph{lifting} the problem, as the recovery of a rank-1, positive semidefinite (PSD) unknown $\brho^* {\brho^*}^H$ from $\bd$. 
\eqref{eq:phaless} become equivalent to realizations under a \emph{linear} measurement model, governed by a \emph{lifted forward map}, $\mathcal{F} : \mathbb{C}^{N \times N} \mapsto \mathbb{C}^M$, where
\begin{equation}
d_m = \langle \ba_m \ba_m^H, \brho^* {\brho^*}^H \rangle_F, \ \ \text{for} \ \ m = 1, \ldots M.
\end{equation}
This reformulation facilitates the use of established tools from low rank matrix recovery theory through convex-relaxations and semidefinite programming \cite{candes2013phaselift, candes2015phase}.
The injectivity and the spectral properties of $\calF$ over rank-1, PSD matrices therefore determine the exact recovery of $\brho^* {\brho^*}^H$ \cite{candes2013phaselift}.

More recently, algorithms that attain performance guarantees by directly operating on the original signal domain ~\cite{bahmani2017PhaseMax, goldstein2018PhaseMax, hand2016elementary} have been introduced to overcome the demanding computational and memory requirements of the lifting-based approaches.
One of the most prominent one is the WF algorithm~\cite{candes2015phase_IEEE}, which minimizes the following functional:
\begin{align}
    \label{eq:calJ} \calJ(\brho) := \frac{1}{2M}\sum_{m = 1}^M |(\ba_m)^H\brho\brho^H\ba_m - d_m|^2.
\end{align}
At the $p^{th}$ iteration step, the WF algorithm updates the current estimate $\brho^{(p - 1)}$ of the unknown quantity as follows: 
\begin{align}
    \label{eq:gwf_itr_updt} \brho^{(p)} = \brho^{(p - 1)} - \frac{\gamma_p}{\|\brho^{(0)}\|^2}\nabla\calJ(\brho)_{\brho = \brho^{(p - 1)}}.
\end{align}
Here, $\gamma_p$ denotes the learning rate at the $p^{th}$ stage and the gradient is given by the Wirtinger derivative of $\calJ(\brho)$,
\begin{align}
    \nabla\calJ(\brho) & = \left(\frac{\partial\calJ}{\partial\brho}\right)^H.
\end{align}
The critical component of the WF framework is at the initialization step, where $\brho^{(0)}$ is determined from the leading eigenvector $\mathbf{v}_0$ of the backprojection estimate $\bY\in\bbC^{N \times N}$ as follows:
\begin{align}
    \label{eq:spectral_matrix} \bY & := \frac{1}{M}\calF^H(\bd), \\
    \brho^{(0)} & = \sqrt{\lambda_0}\mathbf{v}_0,
\end{align}
where $\calF^H$ is the $\ell_2$ adjoint of $\calF$ and the scaling factor $\sqrt{\lambda_0}$ is a norm-estimate of the unknown image of interest.
We refer to $\bY$ as the spectral matrix.

Under the following concentration inequality on $\bY$
\begin{align}
    \label{eq:cond_phaseless} \|\bY - \left(\brho\brho^H + \|\brho\|^2\mathbf{I}\right)\| \leq \delta\|\brho\|^2,
\end{align}
the initial estimate provably enters a basin of attraction in the neighborhood of the global solution set $\mathbb{P}:= \{ \brho^* \mathrm{e}^{\mathrm{i} \phi}, \phi \in [0, 2 \pi) \}$, such that convergence is guaranteed under the validity of a \emph{regularity condition} for the loss functional $\mathcal{J}$ in the noise-free setting with Gaussian sampling, and coded-diffraction models \cite{candes2015phase}. 
These amount to exact recovery guarantees in the statistical setting, where any $\brho \in \mathbb{C}^N$ can be exactly recovered up to a global phase factor, with overwhelming probability if the number of samples exceeds $\mathcal{O}(N \log N)$.

On the other hand in \cite{yonel2020deterministic}, the validity of \eqref{eq:cond_phaseless} for all $\brho \in \bbC^N$ with a sufficiently small $\delta$ ($<0.184$) was shown to be a sufficient condition for universal exact recovery via WF for any $\calF$ in a deterministic mathematical framework.
Hence, deterministic forward maps, $\bF$, that relate to underlying data collection geometry are equipped with exact recovery guarantees. 
This is especially useful for wave-based imaging applications, where the sampling vectors, $\{ \ba_m \}_{m = 1}^M$, are related to the transmitter and receiver locations, transmission signal waveform, and its speed within the propagation medium, and are unlikely to follow i.i.d. Gaussian distribution.

\vspace{-0.1in}
\subsection{WF with a Deep Decoding Prior}
\label{subsec:WF_with_decoder}
In this paper, we build on the mathematical arguments introduced in~\cite{yonel2020deterministic} in establishing the exact recovery guarantee for a DL-based algorithm.
This allows our DL-based algorithm and theoretical results to be applicable to a wide range of practical imaging applications involving deterministic forward maps.
In particular, we present our phaseless imaging approach that performs WF iterations in \emph{a lower dimensional encoded space} in $\mathbb{C}^{N_y}$, where $N_y \ll N$, in lieu of the original image domain in $\mathbb{C}^N$.

The key distinction from existing phase retrieval theory arises from the \emph{non-linearity} of the underlying measurement map prior to loss of phase information, since \eqref{eq:phaless} corresponds to $\bd = | \bF \brho^* |^2$, where $| \cdot |$ denotes element-wise absolute-value operation, and $\bF \in \mathbb{C}^{M \times N}$ is the matrix with $\{\ba_m^H\}_{m = 1}^M$ as its rows.   
Namely, we now assume that our image class of interest resides in a low dimensional manifold $\mathbb{T}$, embedded in the high dimensional space in $\bbC^N$. 
We aim to capture this image manifold $\bbY$ by parameterization over the $\bbC^{N_y}$ in the range of a non-linear transformation $\calH : \bbY \subset \bbC^{N_y} \mapsto \bbT$,
which we refer to as the \emph{decoder}.
This yields a measurement model of the form:
\begin{equation}\label{eq:phaless2}
d_m = | \langle \ba_m , \calH (\by^*) \rangle |^2, \quad \text{for} \ \ m = 1, \ldots M 
\end{equation}
where $\brho^* = \calH (\by^*)$, such that we have a compositely non-linear mapping, $\bd = | \bF \calH (\by^*) |^2$, over the low dimensional parameter space in $\bbY \subset \bbC^{N_y}$. 

The problem consists of two key elements: \emph{i}) given $\calH$, solving for the underlying, \emph{compressive} representation $\by \in \bbY$ from \eqref{eq:phaless2}, and \emph{ii}) solving for an $\calH$ that sufficiently approximates the image manifold $\bbT \subset \bbC^N$. 
While the first component requires the composite mapping formed by $\bF$ and $\calH$ to demonstrate favorable properties of the parameter space, the other requires constructing one such representation in the first place. Practically, the two can be summarized under an objective using a training set of $\bbD := \{ \brho^*_t, \bd_t \}_{t = 1}^T$, such that
\begin{align}\label{eq:problStat}
    \underset{\{ \by_t \}_{t =1}^T, \calH \in \bbW}{\text{arg min}} &\  \frac{1}{T M} \sum_{t = 1}^T \sum_{m = 1}^M  | \ba_m^H \calH(\by_t) \calH(\by_t)^H \ba_m - d_{t,m} |^2 \\
    s.t. \ \ & \| \calH(\by_t) - \brho^*_t \| \leq \varepsilon, \ \forall t = 1, \ldots T, \nonumber
\end{align}
where $\bbW$ denotes a space of functionals that acts as a constraint in the search of $\calH$, and $\varepsilon > 0$ models the approximation error in the range of the decoder.

Ultimately, despite serving as a conceptual motivation, solving \eqref{eq:problStat} is not meaningful without attaining proper generalization over the image manifold $\bbT$, i.e., any $\brho \in \bbT$ must be reliably reconstructed by recovering its encoded representation from its intensity-only measurements. 
To this end, we enlist a DL-based approach, where $\calH$ is obtained in a \emph{task-driven} manner, such that it facilitates the accurate recovery of elements in $\bbT$ in its range after the iterative procedure of WF is deployed on the lower dimensional, encoded parameter space.
The DL-based approach effectively \emph{splits} the objective in \eqref{eq:problStat} to be minimized over its forward, and back-propagation stages.
Namely, at the forward pass, we pursue a solution
$\hat{\by} \in \bbY$
that minimizes the following objective function for each training sample: 
\begin{align}\label{eq:calK}
    {\calK}(\by) := \frac{1}{2M}\sum_{m = 1}^M \left[(\ba_m)^H\calH(\by)\calH(\by)^H\ba_m - d_m\right]^2,
\end{align}
whereas in the back-propagation, we use the solution $\hat{\by}$ to formulate the training loss over $\calH \in \bbW$, evaluated over the training set $\bbD$.

Accordingly, our approach incorporates \emph{a deep decoding prior} into the WF framework.
Deep decoding prior refers to the type of compressive representation implemented under our decoding network $\calH$, as it constrains the reconstructed images to its output space.
We opt to use decoding prior in referring to $\calH$ to differentiate our overall approach from works that consider generative priors~\cite{hand2018phase, shamshad2018robust, hand2020compressive, shamshad2020robust}, which do not utilize the phaseless measurements and the corresponding ground truth images for training, and instead use pre-trained GAN generator model
for $\calH$. 
The end-to-end training of $\calH$ transforms an $M/N$ phase retrieval problem into an $M/N_y$ phase retrieval problem akin to the generative prior setting. 
Hence, the composite operator mapping $\by$ to the measurements attains a higher oversampling factor, albeit, at the cost of non-linearity. 
On the other hand, overcoming the $N/N_y$ factor reduction is offloaded to the approximation capability of the decoder. In accordance, we are interested in the theoretical justifications of recovering a true representation $\by^* \in \bbY$, for a given $\calH$ such that $\calH(\by^*) = \brho^*$, using the iterative scheme of WF.


Unlike \cite{hand2018phase, hand2020compressive}, our architecture is based on the observation that $\calH$ and the measurement map $\calF$ need to satisfy certain sufficient conditions for exact recovery \emph{in composition with each other}.
This serves as our key motivation to utilize end-to-end training, as it directly entangles the presence of the generator with the measurement map of the problem, hence drives the training procedure to enhance the feasibility of the phase retrieval problem over $\bbT$.
However, guarantees on finding such an $\calH$, or the impact of approximation and generalization errors encountered in the training of $\calH$ are beyond the scope of this paper.

\section{Network Architecture}
\label{sec:DN}
\begin{figure}[!t]
\centering
\subfloat[]{\includegraphics[width=0.99\columnwidth]{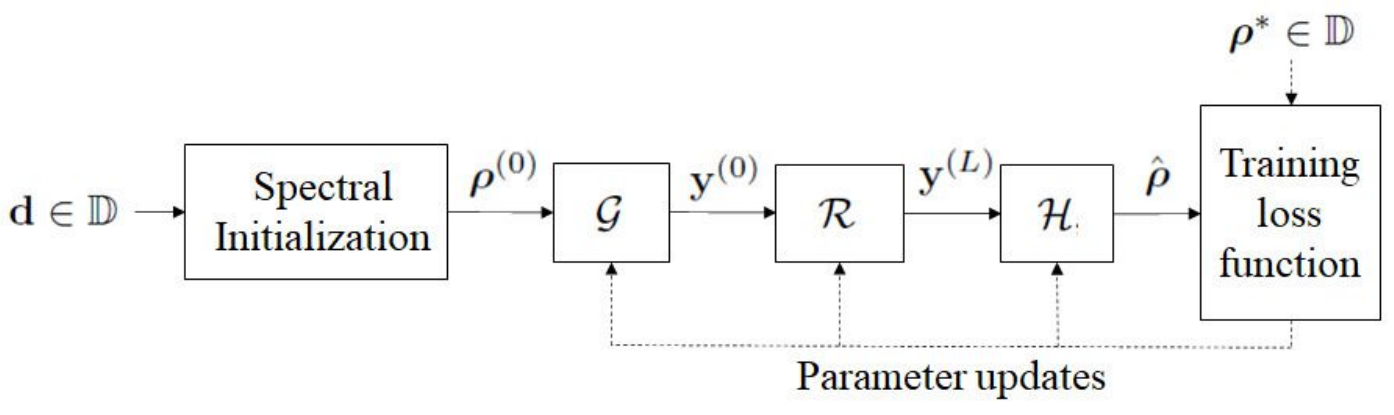}%
\label{fig_training_phase}}
\\
\subfloat[]{\includegraphics[width=0.9\columnwidth]{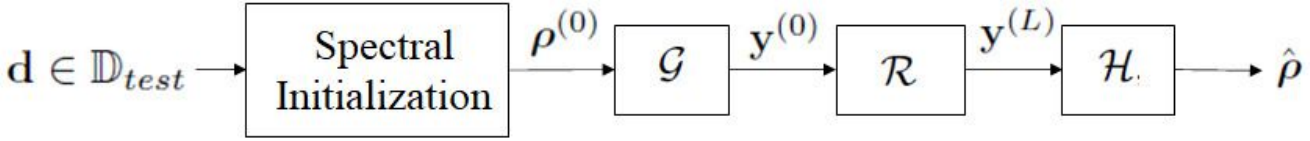}%
\label{fig_test_phase}}
\caption{Schematic diagrams showing the (a) training and (b) inversion processes.}
\label{fig:training_test_phase}
\end{figure}
As our phaseless imaging approach recovers an encoded version of the unknown image through WF updates, our first challenge is to design an efficient initialization scheme for the encoded image space.
To this end, we utilize the spectral initialization step, and learn a non-linear transformation from the set of initial estimates $\bbS\subseteq\bbC^N$ to the set of encoded initial estimates $\bbY_0\subseteq\bbC^{N_y}$, $\calG:\bbS\mapsto\bbY_0$, to map the spectral estimate $\brho^{(0)}\in\bbS$ to an initial estimate $\by^{(0)}\in\bbY_0$ in the encoded image space.
We refer to $\calG$ as the \emph{encoding network}.
We use an $L$-layer RNN, $\calR$, to generate the final estimated encoded image $\calR(\by^{(0)}) = \by^{(L)} = \hat{\by}$, where $\by^{(L)}$ is the output of the $L^{th}$ layer of the RNN.
We denote the set of encoded image values generated at the $l^{th}$ RNN layer by $\bbY_l\subset\bbC^{N_y}$ for $l\in[L - 1]$ and define $\bbY$ as $\bbY = \bigcup_{l = 0}^L\bbY_l \subset \bbC^{N_y}$.
Thus, $\calR:\bbY_0\mapsto\bbY$.
Finally, the output from the RNN is decoded back by $\calH:\bbY\mapsto\bbT$ to generate the estimated image $\hat{\brho}\in\bbT$.
Under exact recovery, $\hat{\brho} = \brho^*$.

In our network architecture, the encoder, RNN and the decoder are jointly learned through supervised training.
The training dataset $\bbD$ is composed of different ground truth or correct images and the corresponding intensity measurement vectors.
On the other hand, each new sample from the test set, $\bbD_{test}$, only requires the intensity measurement vector which is then applied to the trained imaging network to produce the estimated image.
A block diagram of the training and inversion phases of our algorithm are shown in Fig.~\ref{fig_training_phase} and~\ref{fig_test_phase}, respectively.

\subsection{RNN Structure from the Iterative WF Updates}
\label{subsec:objective}
Starting from the initial encoded representation $\by^{(0)}$, iterative WF update at the $l^{th}$ stage is calculated as follows:
\begin{align}
    \label{eq:yl_y_l_minus_1} \by^{(l)} & = \by^{(l - 1)} - \frac{\gamma_{l}}{\|\by^{(0)}\|^2}\nabla\calK(\by)_{\by = \by^{(l - 1)}}.
\end{align}
$\by^{(l)}$ denotes the output at the $l^{th}$ iteration and $\gamma_{l}$ is a positive real-valued constant associated with the learning rate for the $l^{th}$ update.
The WF update in~\eqref{eq:yl_y_l_minus_1} results in $\by^{(l)}$ that reduces the data fidelity term $\calK(.)$ compared to $\by^{(l - 1)}$.
The gradient of $\calK(\by)$ with respect to $\by\in\bbC^{N_y}$ is given by
\begin{align}
    \label{eq:nablaJ} \nabla \calK(\by) & = \left(\frac{\partial\calK}{\partial\by}\right)^H = \frac{1}{M}\nabla\calH(\by)\calF^H \left(\be\right)\calH(\by),
\end{align}
where $\be = \begin{bmatrix}e_1 & \cdots & e_M\end{bmatrix}$ and
$e_m\in\mathbbm{R}$, for $m\in[M]$, is defined as $e_m := \ba^H_m\calH(\by)\calH(\by)^H\ba_m - d_m$.

\begin{figure*}
\centering
\includegraphics[width=1.2\columnwidth]{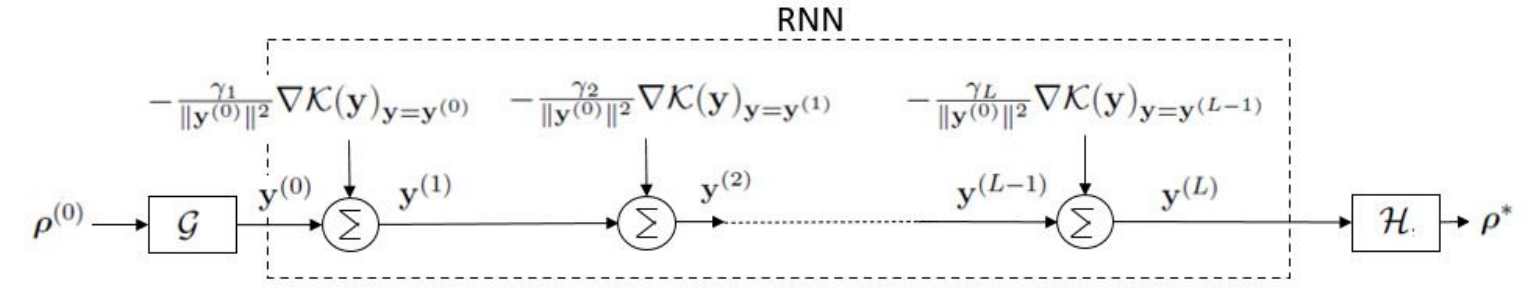}
\caption{Block diagram of DL-based phaseless imaging network.}
\label{fig:network_WF_DL}
\end{figure*}
Instead of continuing to update the encoded representation until convergence, we consider a fixed number of iterative update steps over which the algorithm is promoted to recover accurate solutions over certain conditions on the network parameters.
Similar to~\cite{mason2017_radar, yonel2017deep, yonel2019deep, kazemi2019deep, kazemi2020deep}, $L$ number of subsequent update steps from~\eqref{eq:yl_y_l_minus_1} are mapped into the stages of an $L$-layer RNN.
The resulting network is referred to as an RNN due to the recursive nature of its architecture.
Each RNN layer essentially carries out a WF update on the encoded representation.
The learning rate related constants, $\{\gamma_l\}_{l = 1}^L$, are all trainable parameters of the RNN whose values are learned during the training process.
The overall diagram of our DL-based inversion network for phaseless imaging is shown in Fig.~\ref{fig:network_WF_DL}.

\vspace{-0.13in}
\subsection{Lipschitz Constants of the DL Networks}
\label{subsec:lipschitz_constant_def}
The encoding and decoding prior networks are trained with the goal of recovering images from their low dimensional representations at a faster convergence rate compared to the WF algorithm and extending the recovery guarantees of~\cite{yonel2020deterministic} to challenging problem settings with $M$$<$$N$ for arbitrary forward maps.
In order to achieve the above two objectives, we characterize the impact of the encoder and decoder networks on recovery guarantees through their Lipschitz constants, rather than the explicit architectures of the networks or any probabilistic properties on their learned parameter values.
Appropriate ranges of these constants that are associated with improved recovery performance compared to the WF algorithm is presented in Section~\ref{sec:recovery_guarantees}.

The Lipschitz constant of $\calG$ is defined as the smallest value of $\mu_{\calG}\in\mathbbm{R}^+$ satisfying~\cite{miyato2018spectral}
\begin{align}
    \label{eq:M_G_Lip} \frac{\|\calG(\brho^{(0)}_1) - \calG(\brho^{(0)}_2)\|}{\|\brho^{(0)}_1 - \brho^{(0)}_2 \|} \leq \mu_{\calG},
\end{align}
$\forall\brho^{(0)}_1,\brho^{(0)}_2\in\bbS$, and is given by
\begin{align}
    \mu_{\calG} = \sup_{\brho^{(0)}\in\bbS} \sigma(\nabla\calG(\brho^{(0)})),
\end{align}
where $\sigma(\bA)$ denotes the largest singular value of $\bA$.

Suppose $\by^{(0)} = \calG(\brho^{(0)})$ is expressed as a function of a set of weight matrices $\bU_j\in\bbC^{P_j \times P_{j - 1}}$'s, bias vectors $\bb_j\in\bbC^{P_j}$'s, and non-linear functions $f_j(.)$'s, where $j\in[J]$, with $P_0 = N$ and $P_J = N_y$.
The output at the $j^{th}$ step, denoted by $\breve{\by}_j\in\bbC^{P_j}$, relates to its input $\breve{\by}_{j - 1}\in\bbC^{P_{j - 1}}$ as follows:
\begin{align}
    \breve{\by}_j & = f_j(\bU_j\breve{\by}_{j - 1} + \bb_j),
\end{align}
where $\breve{\by}_J = \by^{(0)}$ and $\breve{\by}_0 = \brho^{(0)}$.
The activation function $f_j(.)$ operates componentwise on the corresponding vector inputs.
For the choice of $f_j(.)$ as the rectified linear unit (ReLU), Lipschitz constant of $f_j(.)$ is upper-bounded by $1$
and thus, the Lipschitz constant of $\calG(.)$ for this case is upper bounded by $\prod_{j = 1}^J\sigma(\bU_j)$.
Similarly, the Lipschitz constant of $\calH$ is calculated as $\mu_{\calH} = \sup_{\by\in\bbY} \sigma(\nabla\calH(\by))$.

\section{Theoretical Foundations}
\label{sec:convergence}
In order to justify the effectiveness of our phaseless imaging approach, we provide a theoretical foundation towards attaining exact recovery for a given, arbitrary lifted forward map $\mathcal{F}$, and an image manifold that is assumed to be characterized in the range of a non-linear operator $\calH$.

In terms of the technical content of the exact recovery theory, our work differs from prior works in \cite{hand2018phase, hand2020compressive} in two notable ways.
The first pertains to the conditions exerted on $\calH$.
In \cite{hand2018phase, hand2020compressive}, a pre-determined architecture is assumed for $\calH$ and a concentration property on the network weights is used to facilitate recovery guarantees by a sufficient condition on $\calF$. 
We do not deploy an architecture specification for $\calH$, and only assume a \emph{local} concentration-type property instead. 
The second pertains to the sufficient condition on the measurement map $\calF$, where
\cite{hand2018phase, hand2020compressive} use a \emph{range restricted} RIP-type property on the underlying linear sampling vectors $\bF$, while our sufficient condition enforces a range restriction on the sufficient condition introduced in \cite{yonel2020deterministic}. 
The major distinction arises in the domain of the accompanying concentration property, where our work evades the requirement of validity over pair-wise differences. 


\subsection{Approach}
To understand the feasibility of such a theoretical justification under an arbitrary pairing of $\calF$ and $\calH$, it is useful to initially revisit the standard phase retrieval problem in the statistical setting of Gaussian sampling. 
Indeed, theoretical results in phase retrieval literature commonly consider this case, where $\ba_m$ are i.i.d. complex Gaussian distributed, with which the recovery from intensity-only measurements is achieved with overwhelming probability \cite{candes2015phase_IEEE, yuan2019SWF}. 

Using the property that Gaussian distribution is invariant under unitary transformations, the classically studied statistical phase retrieval problem under the Gaussian sampling model is equivalent to a $1D$-Fourier phase retrieval problem under a \emph{linear Gaussian generator}:
\begin{equation}
\bd = | \mathbf{A} \bs |^2 = | \mathbf{F}_M \mathbf{F}^H_M \mathbf{A} \bs |^2  = | \mathbf{F}_M \tilde{\mathbf{A}} \bs |^2 = |  \mathbf{F}_M \bt |^2,
\end{equation} 
where $\bA\in\bbC^{M \times {N_s}}$ has all i.i.d. Gaussian distributed components, $\bs \in \mathbb{C}^{N_s}$, $\bF_M\in\bbC^{M \times M}$ is the discrete 1D-Fourier matrix, $\tilde{\bA} = \bF^H_M\bA$ and $\bt = \tilde{\mathbf{A}} \bs$.
In other words, standard statistical theory states that a signal $\bt \in \bbC^M$ realized from a Gaussian generative prior can provably be recovered from its $M-$point periodogram, \emph{if the intrinsic dimension $N_s$ is sufficiently low}. 

Exact recovery guarantees in the statistical setting highlight the power of having \emph{a generative prior} at inference, albeit disguised as the measurement model due to spherical symmetry of the Gaussian distribution. 
This is because the $1D$-Fourier phase retrieval problem is well-known to be {severely} ill-posed: it admits at best $2^M$ non-equivalent solutions in the feasible set of $\bd = |  \mathbf{F}_M \bt |^2$ for an arbitrary $\bt \in \bbC^M$ \cite{jaganathan2017sparse}. 
The linear Gaussian generator alleviates the fundamental limitations in this regard, and provides a guarantee directly on the lower, $N_s$-dimensional encoded space, given that $\bt = \tilde{\mathbf{A}} \bs$. 

Ultimately, our work aims at generalizing this phenomenon by: \emph{i}) using the deterministic setting of \cite{yonel2020deterministic} to account for an arbitrary $\calF$, and \emph{ii}) incorporating the presence of a non-linear $\calH$ that can capture the signal domain. 
To this end, we quantify the impact of operating in the lower dimensional encoded domain on the existing deterministic guarantees of \cite{yonel2020deterministic} by specifying conditions on $\calH$ within the sufficient conditions, and identifying the numerical impact of the generator, i.e. our decoder, on convergence guarantees.



\subsection{Background}
\label{subsec:background}
\paragraph{Exact phase retrieval theory}
For universality of exact recovery described in \cite{yonel2020deterministic},
the concentration bound in~\eqref{eq:cond_phaseless} is a sufficient condition if it holds over all $\brho \in \mathbb{C}^N$ with $\delta < 0.184$. The terms involved in the concentration bound are relevant for the initial estimate to land within a basin of attraction around the true solution $\brho^*$, guaranteeing that:
\begin{equation}\label{eq:rip2}
\frac{1}{M} \| \calF(\brho\brho^H - \brho^*{\brho^*}^H) \|^2 \geq (1 - \delta^{WF}_1) \| \brho\brho^H - \brho^*{\brho^*}^H\|^2_F.
\end{equation}
Let $\bbN_{\epsilon}(\brho^*)$ denote this $\epsilon$-neighborhood of $\brho^*$ obtained from the sufficient condition in \eqref{eq:cond_phaseless}, and $\epsilon\in\mathbbm{R}^+$ and $\delta^{WF}_1\in\mathbbm{R}^+$ are both functions of $\delta$. In the end, \eqref{eq:rip2} facilitates the restricted strong convexity around the solution $\brho^*$ if $\delta^{WF}_1 < 1$, which \eqref{eq:cond_phaseless} guarantees an initial estimate to land in for any $\brho^*$.
The way to establish \eqref{eq:cond_phaseless} as a sufficient condition is through deriving \eqref{eq:rip2} as a deterministic consequence, and showing that the requirement of $\delta^{WF}_1 < 1$ implies $\delta < 0.184$ in the sufficient condition.

\paragraph{Range restriction with $\calH$}
The stringency of the sufficient condition in \eqref{eq:cond_phaseless} arises through its universality over all $\brho \in \bbC^N$ and the corresponding requirement for $\calF^H\calF$ to be well-conditioned over the manifold of rank-1 PSD matrices.
On the other hand, with the presence of $\calH$, the range of the decoder incorporates an additional constraint, and hence, creates a smaller feasible set for the problem over which $\calF^H \calF$ should be well-conditioned instead.

An intuitive incorporation of the image manifold in the recovery guarantees therefore is by restricting the parameter space of the original concentration bound, where the lifted normal operator is to satisfy, for all $\by \in \bbY \subset \bbC^{N_y}$:
\begin{align}
\bigg\| \frac{1}{M} \calF^H \calF(\calH(\by) \calH(\by)^H) &- \left(\calH(\by)\calH(\by)^H + \| \calH(\by) \|^2 \mathbf{I} \right) \bigg\| \nonumber \\
\label{eq:cond_phaseless2} &\leq \delta \| \calH(\by) \|^2. 
\end{align}
\eqref{eq:cond_phaseless2} shows that the concentration property of $\calF^H\calF$ is now required to hold over only the image manifold captured by the range of $\calH$.

To fully understand the usefulness of this condition, we must establish its corresponding restricted strong convexity property over the image manifold. 
Namely, for a $\brho = \calH(\by)$, and a ground truth $\brho^* = \calH(\by^*)$, does \eqref{eq:cond_phaseless2} with a sufficiently small $\delta$ imply the property in \eqref{eq:rip2} with $\delta^{WF}_1$ replaced by ${\delta_1} < 1$ in some locality in the \emph{parameter space}, i.e. $\by \in \bbN_{\epsilon_{\by}}(\by^*)$?
Here, $\bbN_{\epsilon_{\by}}(\by^*)$ denotes the $\epsilon_{\by}$-neighborhood of $\by^*$ and $\epsilon_{\by}\in\mathbbm{R}^+$.

\paragraph{The limitation for sufficiency}

In order to verify whether the restricted concentration property is sufficient, we consider first the linear perturbation operator $\Delta$ that maps $\brho\brho^H$ to $\bbC^{N\times N}$ over all $\brho$ vectors that are reproducible by the decoder from $\by\in\bbY$, as
\begin{align}
\Delta(\brho\brho^H) = \frac{1}{M}\calF^H\calF(\brho\brho^H) - (\brho\brho^H + \|\brho\|^2\bI).
\end{align}
Similarly to the steps of the proof of Lemma III.4 in~\cite{yonel2020deterministic}, it is easy to verify that the validity of the restricted strong convexity condition through~\eqref{eq:rip2} hinges on the concentration property of a perturbation operator $\Delta$, over the pairwise differences,
\begin{equation}\label{eq:DeltaB}
\bigg| \langle \Delta(\bE_{\brho}), \bE_{\brho} \rangle_F \bigg| \leq \delta_1 \ \| \bE_{\brho} \|^2_F,
\end{equation}
where $\bE_{\brho}$ is defined as
\begin{align}
    \label{eq:E_rho_def} \bE_{\brho} & = \calH(\by)\calH(\by)^H - \brho^*{\brho^*}^H,
\end{align}
and that~\eqref{eq:DeltaB} is guaranteed to hold with $\delta_1 < 1$ when \eqref{eq:cond_phaseless2} is satisfied.
As we know, $| \langle \Delta(\bE_{\brho}), \bE_{\brho} \rangle_F |$ can be upper bounded by $\sqrt{2} \| \bE_{\brho} \|_F \| \Delta(\bE_{\brho}) \|$.
Moreover, $\| \Delta(\bE_{\brho}) \|$ can be upper bounded by $ \sum_{i =1}^2 |\lambda_i| \| \Delta(\bv_i \bv_i^H) \| $, where $\lambda_i\in\mathbbm{R}$ and $\bv_i\in\bbC^N$ are the $i^{th}$ eigenvalue and the corresponding eigenvector of $\bE_{\brho}$, respectively, for $i\in\{1, 2\}$.

Consequently, to promote \eqref{eq:cond_phaseless2} as a sufficient condition for our approach, $\bv_i$'s need to be reproducible by the decoding network $\calH$, such that $\| \Delta(\bv_i \bv_i^H) \|$ terms are controlled. 
For an arbitrary pair of $\brho, \brho^*$, the error $\bE_{\brho}$ for the corresponding lifted Kronecker signals admit a direct spectral analysis, such that the $\bv_i$ are formed by \emph{affine combinations} in the range of $\calH$ (see Appendix~\ref{app:eigenvector_lifted_error}). 
This presents the key limitation for the sufficiency of a range restriction by the generator $\calH$, unless the domain of concentration is expanded to include the union of pair-wise affine hulls of the elements in the range of $\calH$.

\subsection{Conditioning $\calH$}
\label{subsec:condition_H}
\paragraph{Sufficiency with linearity}
It is clear that for a linear $\calH$, \eqref{eq:cond_phaseless2} is a sufficient condition, as the affine combinations are reproducible by $\calH$ via an affine combination in the $\bbY$-domain.
However, for a general non-linear $\calH$, the eigenvectors $\bv_i$ do not necessarily admit such a representation.
We instead are interested in casting~\eqref{eq:cond_phaseless2} as a sufficient condition through specific conditions on an arbitrary, non-linear $\calH$.
To this end, we first identify the properties that facilitate our objective when using a linear decoder model, towards obtaining an intuitive extension onto the general case. 
%
The assumption that $\calH$ is a linear map, i.e., $\calH(\by) = \bH\by$ where $\bH\in\bbC^{N\times N_y}$, leads to
\begin{equation}
\label{eq:Delta_Ep1} \| \Delta(\bE_{\brho}) \| = \| \Delta(\bH(\by\by^H - \by^*{\by^*}^H)\bH^H) \|.
\end{equation}
Now, $\by\by^H - \by^*{\by^*}^H$ can be represented by its eigenvalues and eigenvectors as $\sum_{i = 1}^2 \lambda_i \bu_i \bu_i^H$, where $\lambda_i\in\mathbbm{R}$ and $\bu_i\in\bbC^{N_y}$ are eigenvalues and the corresponding eigenvectors for $i=1,2$.
$\bu_1$ and $\bu_2$ are constructed from affine combinations of $\by$, $\by^*$ per spectral analysis presented in Appendix~\ref{app:eigenvector_lifted_error}.

\paragraph{Requirements for the general case}
We now assume that~\eqref{eq:cond_phaseless2} holds for all $\calH(\by)$, $\by\in\bbR^{N_y}$ for convenience.
Therefore, since $\| \Delta(\bE_{\brho}) \|$ can be upper bounded by $\sum_{i = 1}^2 |\lambda_i | \| \Delta( (\bH \bu_i) ( \bH \bu_i)^H ) \|$ when $\calH$ is linear, then using the relation in~\eqref{eq:cond_phaseless2}, we have from~\eqref{eq:Delta_Ep1},
\begin{equation}
\| \Delta(\bE_{\brho}) \| \leq \delta \sum_{i = 1}^2 |\lambda_i |  \| \bH \bu_i \|^2.
\end{equation}
Here, the first crucial property of $\calH$ arises, as the Lipschitz continuity of $\calH$, along with the assumption that $\calH(\bzero) = \bzero$, which yields the following upper bound for linear $\calH$:
\begin{equation}\label{eq:Delta_Ep3}
\begin{split}
\| \Delta(\bE_{\brho}) \| &\leq \delta \ \text{max}(|\lambda_1|, |\lambda_2|) \sum_{i = 1}^2 \| \bH \bu_i \|^2 \\
&\leq 2 \delta \mu^2_{\calH}  \| \by \by^H - \by^* {\by^*}^H \|. 
\end{split}
\end{equation}

Although this bound is not the tightest, it is of interest because, it gives a blueprint that befits generalization to the non-linear setting. 
Mainly, in the linear setting with a \emph{spectrally well-conditioned} generator, we can obtain a universal constant ($2$ in this case) that upper bounds this perturbation operator only through the leading eigenvalue-eigenvector pair, since $ \| \mathbf{H} \bu_i \|^2  \leq \mu^2_{\calH}$ by the Lipschitz property of $\mathbf{H}$.
The key observation is that via an encoder-decoder scheme that enforces the model to operate in an $\epsilon_{\by}$-neighborhood in the parameter space, such a condition as in \eqref{eq:Delta_Ep3} is only needed to be satisfied \emph{locally} over $\bbY$, in lieu of the global property demonstrated by a linear $\calH$. 

\paragraph{Extension via a local property}
For a general non-linear decoder, we instead perform this analysis using an operator $\tilde{\calH}: \bbC^{N_y \times N_y} \mapsto \bbC^{N \times N}$, which is defined as follows:
\begin{enumerate}
\item Given input $\bZ \in \bbC^{N_y \times N_y}$, extract the leading eigenvalue-eigenvector pair: $\lambda_0$, $\bu_0$.
\item Apply $\calH$ on $\sqrt{\lambda_0} \bu_0$ to calculate $\calH( \sqrt{\lambda_0} \bu_0 )$. 
\item Get output $\tilde{\calH}(\bZ)$ by lifting: $\calH ( \sqrt{\lambda_0} \bu_0 ) \calH ( \sqrt{\lambda_0} \bu_0 )^H$.
\end{enumerate}
Under this definition, our desired bound on the perturbation operator for a generic $\calH$ as can be written as
\begin{equation}\label{eq:finalCond}
\| \Delta(\tilde{\calH}(\by\by^H) - \tilde{\calH}(\by^*{\by^*}^H)) \| \leq \hat{\delta} \| \by\by^H - \by^*{\by^*}^H \|_F,
\end{equation} 
which, incorporating the locality property on the encoded domain, should hold $\forall \by^* \in \bbY$, $\by \in \bbN_{\epsilon_{\by}}(\by^*)$.
For the PSD rank-1 inputs $\by\by^H$ and $\by^*{\by^*}^H$, $\tilde{\calH}(\by\by^H)$ and $\tilde{\calH}(\by^*{\by^*}^H)$ are equal to $\calH(\by)\calH(\by)^H$ and $\brho^*{\brho^*}^H$, respectively.
Moreover, we are not necessarily interested in this bound globally as obtained for the linear case in~\eqref{eq:Delta_Ep3}, but only locally, since that is sufficient for our guarantees.

This leads us to the following property on $\calH$: for the definition of $\tilde{\calH}$ presented above, for a given $\calF$, the following inequality is satisfied $\forall \by^* \in \bbY$, $\by \in \bbN_{\epsilon_{\by}}(\by^*)$:
\begin{equation}\label{eq:Hcondition}
\| \Delta(\tilde{\calH}(\by\by^H) -  \tilde{\calH}(\by^*{\by^*}^H)) \| \leq \omega(\epsilon_{\by}) \| \Delta(\tilde{\calH}(\by\by^H - \by^*{\by^*}^H)) \|,
\end{equation}
where $\omega(\epsilon_{\by})$ is a positive real-valued constant.
We omit the term in the bracket for future references to this constant, and its dependency on $\epsilon_{\by}$ should be understood.
Under this condition, it is straightforward to verify that the desired bound in~\eqref{eq:finalCond} is satisfied with a constant
\begin{align}
    \hat{\delta} & = \omega\mu^2_{\calH}\delta,
\end{align}
as shown in Appendix~\ref{app:derivation_26_27}.

\section{Recovery Guarantees}
\label{sec:recovery_guarantees}
In this section, we present the exact recovery guarantee for our end-to-end DL-based algorithm.
This result is built upon the theoretical foundations presented in Section \ref{sec:convergence}. 
We elaborate on the numerical implications of our result, and discuss its key outcomes in quantifying the impact and limitations of incorporating a decoding prior.
Finally, we consider the practical implications of our result for implementation purposes.

\subsection{Main Result}
Let $\mathrm{dist}(\by^{(0)}, \by^*)$ be the distance between $\by^{(0)}$ and $\by^*$ defined as follows:
\begin{align}
    \label{eq:dist_definition} \mathrm{dist}(\by^{(0)}, \by^*) & = \min_{\phi\in[0, \pi]}\|\by^{(0)} - \by^*\rme^{i\phi}\|.
\end{align}
Our main result concerns the convergence of the WF iterates to the true representation in the encoded space via our unrolled, encoder-decoder network architecture.
Let $\mu_{\calG}$, $\mu_{\calR}$ and $\mu_{\calH}$ be the Lipschitz constants of $\calG$, $\calR$ and $\calH$, respectively.
We assume that there exists $\tilde{\mu}_{\calH} > 0$ and $\mu_{\calH} > 0$ such that
\begin{align}\label{eq:LipH}
    \tilde{\mu}_{\calH} \leq \frac{\|\calH(\by_1) - \calH(\by_2)\|}{\|\by_1 - \by_2\|} \leq  \mu_{\calH}, 
\end{align}
for all $\by_1, \by_2 \in \mathbb{Y}$.

We define $\epsilon_{\by}$, which we introduced in Subsection~\ref{subsec:background}, as $\epsilon_{\by} := \chi\mu_{\calH}\epsilon$.
$\chi$ is a positive real-valued constant and $\epsilon$ is defined in (21) in~\cite{yonel2020deterministic}.
$\chi$ is lower bounded by
\begin{align}
 \label{eq:k_ub_lb} \chi & \geq \max\left[b_1(\mu_{\calG}, \mu_{\calH}, \epsilon), b_2(\mu_{\calG}, \mu_{\calH}, \mu_{\calR}, \epsilon)\right],
\end{align}
and $b_1(\mu_{\calG}, \mu_{\calH}, \epsilon)$ and $b_2(\mu_{\calG}, \mu_{\calH}, \mu_{\calR}, \epsilon)$ are defined in Appendix~\ref{app:initialization} along with the detailed derivation of~\eqref{eq:k_ub_lb}.
We also define the following quantities:
\begin{align}
    c(\delta, \epsilon_{\by}) & := (1 + \epsilon_{\by})(2 + \epsilon_{\by})(2 + \omega\delta), \\
    \epsilon_{\brho} & := \mu_{\calG}\mu_{\calR}\mu_{\calH} (1 + \epsilon) \epsilon_{\by}, \\
    \delta_1 & := \frac{\sqrt{2}\hat{\delta}(2 + \epsilon_{\brho})(2 + \epsilon_{\by})}{\tilde{\mu}^2_{\calH}(1 - \epsilon_{\brho})(2 - \epsilon_{\brho})}, \\
    h(\delta, \epsilon_{\by}) & := \tilde{\mu}^4_{\calH}(1 - \delta_1)(1 - \epsilon_{\brho})(2 - \epsilon_{\brho}).
\end{align}

\begin{theorem}
\label{theorem:convergence}
Suppose the conditions in~\eqref{eq:cond_phaseless2} and~\eqref{eq:Hcondition} are satisfied for all $\by\in\bbY$, where $\bbY$ is an affine subset of $\bbC^{N_y}$.
Additionally, assume that there exist $\tilde{\mu}_{\calH} > 0$ and $\mu_{\calH}$ such that~\eqref{eq:LipH} holds; and $\calG(\bzero) = \bzero$ and $\calH(\bzero) = \bzero$.
Then, starting from $\by^{(0)}$ that is $\epsilon_{\by}$-distant from $\by^*$, using the step sizes $\frac{\gamma_l}{\|\by^{(0)}\|^2}\leq\frac{2}{\beta}$, the iterates in~\eqref{eq:yl_y_l_minus_1} satisfy
\begin{align}
    \label{eq:convergence_rate3} \rmdist^2(\by^{(j)}, \by^*) & \leq \epsilon^2_{\by}\left[\prod_{l = 1}^{j}\left(1 - \frac{2\gamma_{l}}{\alpha\|\by^{(0)}\|^2}\right)\right]\|\by^*\|^2,
\end{align}
for $j \in [L]$, where $\alpha, \beta > 0$ are such that
\begin{align}
    \label{eq:4_by_ab}\frac{4}{\alpha\beta} & \leq \left(\frac{\tilde{\mu}_{\calH}}{\mu_{\calH}} \right)^8 \left(\frac{h(\delta, \epsilon_{\by})}{c(\delta, \epsilon_{\by})}\right)^2.
\end{align}
\end{theorem}
\begin{proof}
See Appendix~\ref{app:proof_theorem2}.
\end{proof}

This theorem unveils a number of important implications. 
Most notably, the concentration bound parameter $\delta$ is no longer the sole determinant of the recovery guarantee, as for the regime in \eqref{eq:convergence_rate3} to be valid, several parameters must compositely satisfy the inequality $\delta_1 < 1$.
Once this strict bound is violated, we no longer have a feasible $\alpha, \beta$ pair to guarantee the convergence in the encoded parameter space. This, in turn, requires
\begin{align}
    \label{eq:delta_ub} \delta & < \left(\frac{\tilde{\mu}_H}{\mu_H}\right)^2\frac{(1 - \epsilon_{\brho})(2 - \epsilon_{\brho})}{\sqrt{2}\omega(2 + \epsilon_{\brho})(2  + \epsilon_{\by})},
\end{align}
within our sufficient conditions of exact recovery.
Furthermore, we can infer that $\tilde{\mu}_{\calH}$, which is smaller than $\mu_{\calH}$ by definition, should be away from $0$ and $\epsilon_{\brho}$ should be less than $1$ as both of these constants affect the feasibility of the bound in~\eqref{eq:delta_ub}.

\subsection{Sketch of Proof for Theorem~\ref{theorem:convergence}}
Proof of the exact recovery guarantee in Theorem~\ref{theorem:convergence} depends on achieving an initial encoded image within a small neighborhood of the correct encoded unknown $\by^*\in\bbY$.
For our initialization scheme described in Section~\ref{sec:DN} and under the condition from~\eqref{eq:cond_phaseless2}, we have $\rmdist^2(\brho^{(0)}, \brho^*) \leq \epsilon^2\|\brho^*\|^2$ and
\begin{align}
     \label{eq:dist_y} \rmdist^2(\by^{(0)}, \by^*) & \leq \epsilon^2_{\by}\|\by^*\|^2.
\end{align}
The inequality relation in~\eqref{eq:dist_y} is derived in Appendix~\ref{app:initialization}.
Our regularity condition states that for all $\by\in \bbN_{\epsilon_{\by}}(\by^*)$, $\calK(\by)$ satisfies the following inequality:
\begin{align}
    \label{eq:regularity_cond_main} \rmRe\left(\langle\nabla\calK(\by), \be_{\by}\rangle\right) \geq \frac{1}{\alpha}\|\be_{\by}\|^2 + \frac{1}{\beta}\|\calK(\by)\|^2,
\end{align}
where $\be_{\by} = \by - \by^*$ and $\alpha,\beta > 0$.
This ensures local strong convexity of $\calK(\by)$ within the $\epsilon_{\by}$ neighborhood of $\by^*$.
Under~\eqref{eq:cond_phaseless2} and~\eqref{eq:Hcondition}, the regularity condition~\eqref{eq:regularity_cond_main} is observed to be equivalent to
\begin{align}
    \label{eq:regularity_cond_compare2_main} \frac{1}{\alpha\|\by^*\|^2} + \frac{1}{\beta}\mu^8_{\calH}c^2(\delta, \epsilon_{\by})\| \by^*\|^2 & \leq h(\delta, \epsilon_{\by}).
\end{align}
Therefore, for~\eqref{eq:regularity_cond_main} to be satisfied by $\calK(\by)$ for all $\by\in\bbN_{\epsilon_{\by}}(\by^*)$, the left hand side of~\eqref{eq:regularity_cond_compare2_main} is required to be smaller than $h(\delta, \epsilon_{\by})$, which, in turn, leads to the condition in~\eqref{eq:4_by_ab}.
Finally, by expanding $\|\by^{(l)} - \by^*\|$ using~\eqref{eq:yl_y_l_minus_1}, and through~\eqref{eq:regularity_cond_main} and the upper bound $\frac{2}{\beta}$ on the step sizes, we arrive at the result in~\eqref{eq:convergence_rate3}.

\subsection{Key Outcomes}
\paragraph{Implications on the rate of convergence}
By using fixed step sizes $\gamma\in\mathbbm{R}^+$ for the $L$ updates and by defining $\gamma' = \frac{\gamma}{\|\by^{(0)}\|^2}\leq\frac{2}{\beta}$, we observe from~\eqref{eq:convergence_rate3} that $\frac{2\gamma'}{\alpha}$ is a convergence rate related term where the convergence rate increases with an increase in its value.
Furthermore, from Theorem~\ref{theorem:convergence}, by using the upper bound $\frac{2}{\beta}$ on the step sizes, we can upper bound $\frac{2\gamma'}{\alpha}$ by $\frac{4}{\alpha\beta}$.
Therefore, we can infer from~\eqref{eq:4_by_ab} that $\frac{h^2(\delta)}{\mu^8_{\calH}c^2(\delta, \epsilon_{\by})}$ is essentially an upper bound on $\frac{2\gamma'}{\alpha}$.
As long as ${\tilde{\mu}_{\calH}}/{\mu_{\calH}}$, $\epsilon_{\by}$, $\epsilon_{\brho}$ and $\omega$ values are such that our modified upper bound on $\frac{2\gamma'}{\alpha}$ is larger than the one for the WF algorithm, our DL based approach will converge faster to the correct solution.

\paragraph{Conditions on the Lipschitz constants}
From the definitions of $\epsilon_{\by}$ and $\epsilon_{\brho}$, it is evident that with $\chi$ equal to $\tau\in\mathbbm{R}^+$, upper bounding $\tau\mu_{\calH}$ and $\tau\mu_{\calG}\mu^2_{\calH}\mu_{\calR}(1 + \epsilon)$ by $\xi_{\by}\in\mathbbm{R}^+$ and $\xi_{\brho}\in\mathbbm{R}^+$, respectively, leads to the upper bound $\epsilon\xi_{\by}$ on $\epsilon_{\by}$ and $\epsilon\xi_{\brho}$ on $\epsilon_{\brho}$.
It is shown in Appendix~\ref{app:proof_42_44} that,
$\tau\mu_{\calH} \leq \xi_{\by} \leq 1$
and $\tau\mu_{\calG}\mu^2_{\calH}\mu_{\calR}(1 + \epsilon) \leq \xi_{\brho} \leq 1$, if
\begin{align}
    \label{eq:MGMH_ub_lb} \frac{(1 - \tau\epsilon\mu_{\calH})}{(1 + \epsilon)} \leq \mu_{\calG}\mu_{\calH} & \leq \min\left[2 - \frac{1}{\mu_{\calR}}, \frac{\xi_{\brho}}{\xi_{\by}}\right]\frac{1}{(1 + \epsilon)}, \\
    \mu_{\calH} & \leq \xi_{\by}/\tau, \\
    \label{eq:muR_ub} \mu_{\calR} & \leq 1.
\end{align}
These bounds are sufficient for upper bounding $\epsilon_{\by}$ by $\epsilon\xi_{\by}$ and $\epsilon_{\brho}$ by $\epsilon\xi_{\brho}$.
For a given $\tau$ and $\omega$, if
\begin{align}
    \label{eq:delta1_less_delta1_cond} \omega\left(\frac{\mu_{\calH}}{\tilde{\mu}_{\calH}}\right)^2\frac{(2 + \epsilon_{\brho})(2 + \epsilon_{\by})}{(1 - \epsilon_{\brho})(2 - \epsilon_{\brho})} \leq \frac{(2 + \epsilon)}{\sqrt{(1 - \epsilon)(2 - \epsilon)}},
\end{align}
then our exact recovery guarantee is valid over a larger range of $\delta$ compared to the WF algorithm.

\paragraph{Requirements on the $\bbY$-domain}
In the theorem statement, we assume that $\bbY$ is an affine subset of $\bbC^{N_y}$.
This assumption is made for mere convenience to deal the fact that the two eigenvectors of $\mathbf{E}_{\by} := \by \by^H - \by^* {\by^*}^H$ are formed by normalized affine combinations of $\by$ and $\by^*$.
This can be verified by following similar steps as the spectral analysis presented in Appendix~\ref{app:eigenvector_lifted_error}.
For contractions in the parameter domain, the concentration property we imply via the $\calH$-condition is required to hold over these eigenvectors, hence, we require that $\bbY$ is an affine set, such that $\bu_1 \in \bbY$. 
Furthermore, this requirement can actually be relaxed to instead involve a \emph{union of subspaces} model for $\bbY$, since we merely need the union of pair-wise affine combinations of these elements $\by, \by^* \in \bbY$.

This yields an interesting premise if the representations pursued for our image manifold are constrained to be \emph{sparse} in the parameter space in $\bbC^{N_y}$. To this end, a $k-$sparsity constraint on representations results in the union of all $2k$-dimensional subspaces in $\bbC^{N_y}$ for $\bbY$. Such a constraint however, must be enforced in the architecture via \emph{projection} operators in the definition of the RNN-module. 
In our architecture and implementations, we do not provide any additional structure in $\bbY$, and simply assume validity over all $\bbC^{N_y}$.

\paragraph{Spectral conditioning of $\calH$}
For convenience in presenting the theoretical results, we assume a global upper and lower Lipschitz property on $\calH$ in \eqref{eq:LipH}. However, once an $\epsilon_y$-neighborhood is guaranteed in the parameter space, it suffices that such a property is needed only locally over the neighborhood of a $\by^*$. To follow through with this relaxation, we need an additional spectral conditioning on $\calH$, such that:
\begin{align}
\label{eq:RIPH} \tilde{\sigma}_{\calH} \| \by \| \leq \| \calH (\by) \| \leq \sigma_{\calH} \| \by \|,
\end{align}
for all $\by \in \mathbb{Y}$ where $\tilde{\sigma}_{\calH}, \sigma_{\calH}\in\mathbbm{R}^+$.
This is the basic premise of assuming that $\calH$ is a \emph{frame} over $\mathbb{Y}$.
In this setting, the recovery guarantees promptly feature both ratio of $\mu_H$ and $\tilde{\mu}_H$, and the ratio of the frame coefficients, where the convergence bound becomes
\begin{align}
    \label{eq:4_by_ab_v2}\frac{4}{\alpha\beta} & \leq \left(\frac{\tilde{\mu}_{\calH}}{\mu_{\calH}} \right)^4 \left(\frac{\tilde{\sigma}_{\calH}}{\sigma_{\calH}} \right)^4 \left(\frac{h(\delta, \epsilon_{\by})}{c(\delta, \epsilon_{\by})}\right)^2,
\end{align}
with the sufficient condition
\begin{align}
    \delta & < \left(\frac{\tilde{\sigma}_{\calH} \tilde{\mu}_{\calH}}{\sigma^2_{\calH}}\right) \frac{(1 - \epsilon_{\brho})(2 - \epsilon_{\brho})}{\sqrt{2}\omega(2 + \epsilon_{\brho})(2  + \epsilon_{\by})}.
\end{align}

Most notably, with a linear $\calH$, if \eqref{eq:RIPH} is satisfied over $\mathbb{C}^{N_y}$, all the ratios reduce to that of frame coefficients.
This is highly relevant for the Gaussian linear encoder, which is the fundamental case that inspired our formulation under an arbitrary decoder. Namely, an over-determined Gaussian matrix satisfies the RIP over the whole domain in $\mathbb{C}^{N_y}$, with the RIP-constant $\delta_{\calH}\in\mathbbm{R}^+$ approaching $0$ as $M/N_y$ (i.e., the oversampling factor) grows, which increasingly well-conditions the problem, consistent with the statistical theory of phase retrieval.

\section{Training}
\subsection{Implementation of Lipschitz Constant Bounds}
\label{subsec:implement_Lipschitz_bound}
For our training set $\bbD$, let the intensity measurement vector and the associated ground truth image for the $t^{th}$ sample, where $t\in[T]$, be denoted by $\bd_t$ and $\brho^*_t$, respectively.
Training loss is computed as the average $\ell_2$-norm difference between the estimated and the ground truth images.
Moreover, since the image estimation $\brho^{(l)}_t$, calculated as $\calH(\by^{(l)}_t)$ at the $l^{th}$ RNN stage, is expected to get gradually closer to $\brho^*_t$ as $l$ increases, an additional term is typically added to the training loss function that sums the average $\ell_2$-norm differences between $\brho^{(l)}_{t}$ and $\brho^*$.
Our training loss $c_{tr}(\bbU)$ is defined as
\begin{align}
     \label{eq:training_loss} c_{tr}(\bbU) & = \frac{1}{T}\sum_{t = 1}^T\left[\|\hat{\brho}_{t} - \brho^*_t\|^2 + \sum_{l = 1}^{L}\eta_l\|\calH(\by^{(l - 1)}_{t}) - \brho^*_t\|^2\right] \nonumber \\
     & + c_0(\bbU).
\end{align}
$\eta_l\in\mathbbm{R}^+$, where $l\in[L]$, is a constant, $\bbU$ denotes the set of parameters of the overall imaging network, and $c_0(\bbU)$ is used to impose desirable properties on the trained networks.
We set $c_0(\bbU)$ as the sum of $c_i(\bbU)$, where $i\in[4]$, and define $c_i(\bbU)$ in the following discussion.

For imposing the property that $\calG(\bzero) = \bzero$ and $\calH(\bzero) = \bzero$, $c_1(\bbU)$ can be set as $\eta_1\left(\|\calG(\brho)|_{\brho = \bzero}\|^2 + \|\calH(\by)|_{\by = \bzero}\|^2\right)$ where $\eta_1\in\mathbbm{R}^+$.
In order to impose a specific Lipschitz constant value on the RNN, we define $c_2(\bbU)$ as follows:
\begin{align}
    c_2(\bbU) & = \eta_2\left(\max_{t_1, t_2\in[T]}\frac{\|\calR(\by^{(0)}_{t_1}) - \calR(\by^{(0)}_{t_2})\|}{ \|\by^{(0)}_{t_1} - \by^{(0)}_{t_2}\|} - \mu_{\calR}\right)^2,
\end{align}
where $\eta_2\in\mathbbm{R}^+$.
The Lipschitz constants of $\calG$ and $\calH$ can be set to specific values using a similar approach as~\cite{yoshida2017spectral} by first setting $c_3(\bbU)$ and $c_4(\bbU)$ equal to $\eta_3\sum_{j = 1}^J\left(\sigma(\bU_j) - \mu^j_{\calG}\right)^2$ and $\eta_4\sum_{k = 1}^K\left(\sigma(\bW_k) - \mu^k_{\calH}\right)^2$, respectively, where $\eta_3, \eta_4\in\mathbbm{R}^+$, $\prod_{j = 1}^J\mu^j_{\calG} = \mu_{\calG}$ and $\prod_{k = 1}^K\mu^k_{\calH} = \mu_{\calH}$.
$\sigma(.)$ and $\bU_j$ are defined in Subsection~\ref{subsec:lipschitz_constant_def}.
$\bW_k\in\bbC^{Q_k\times Q_{k - 1}}$ is the weight matrix at the $k^{th}$ layer of a similar $\calH$ architecture as the one presented for $\calG$ in Subsection~\ref{subsec:lipschitz_constant_def}, where $k\in[K]$, $Q_0 = N_y$ and $Q_K = N$.
While using the stochastic gradient descent to minimize $c_{tr}(\bW)$, in order to calculate the gradients of $c_3(\bbU)$ and $c_4(\bbU)$, we need to estimate the leading eigenvectors of the different weight matrices of $\calG$ and $\calH$, respectively.
A power method is implemented in~\cite{yoshida2017spectral} where the leading eigenvectors estimated during one training update is reused as the initial vectors for the next update, for which the gradient of $c_{tr}(\bbU)$ is calculated using a different mini-batch from the training set.

\subsection{Computational Complexity}
\label{subsec:complexity}
Computational complexity of our approach depends on the number of RNN stages $L$ as well as the network architectures of $\calG$ and $\calH$.
For linear activation functions for $\calG$ and $\calH$, forward propagations through these networks require $\sum_{j = 1}^{J}P_jP_{j - 1}$ and $\sum_{k = 1}^{K}Q_kQ_{k - 1}$ floating-point operations (FLOP), respectively.
For ReLU activation functions and assuming that each comparison operation requires a single FLOP, an additional $\sum_{j = 1}^{J - 1}P_j + \sum_{k = 1}^{K - 1}Q_k + N_y + N$ FLOPs are carried out.
The output of the $\calH$ network is required to be calculated $L + 1$ times.
For the initial encoded image, we calculate the leading eigenvector of $\bY$, defined in~\eqref{eq:spectral_matrix}, using the power method, and it incurs $O(N^3)$ computational cost.
Calculating $\bF\calH(\by^{(l)})$ and then $\calF(\calH(\by^{(l)})\calH(\by^{(l)})^H)$ requires $\mathcal{O}(MN) + \mathcal{O}(M)$ FLOPS in total.
From $\calF(\calH(\by^{(l)})\calH(\by^{(l)})^H)$, calculating $\frac{1}{M}\calF^H(\be)\calH(\by^{(l)})$ takes another $\mathcal{O}(MN) + \mathcal{O}(M)$ operations.
The error related term $\be$ is defined in Subsection~\ref{subsec:objective} after~\eqref{eq:nablaJ}.
$\calH(\by^{(l)})$ and its gradient $\nabla\calH(\by)|_{\by = \by^{(l)}}$ have updated values at each RNN stage, and the gradient is multiplied by an $N$ length vector requiring an additional $\mathcal{O}(NN_y)$ FLOPS per iteration.
With ReLU activation functions, $\calH(\by^{(l)})$ calculation requires $\sum_{k = 1}^{K}Q_kQ_{k - 1} + \sum_{k = 1}^{K - 1}Q_k + N$ FLOPs.
For calculating the gradient, the derivatives of the non-linear function require $N + \sum_{k = 1}^{K - 1}Q_k$ comparisons while the matrix multiplication part requires $\sum_{k = 1}^KQ_kQ_{k - 1} + \sum_{k = 0}^{K - 2}Q_kQ_{k + 1}Q_{k + 2}$ additional FLOPs.
$M$ is typically some constant multiple of $N$, where the constant is significantly smaller than $N$.
If the value of $Q_k$, for $k\in[K - 1]$, are in the order of $N$, then the computational complexity increases to $\mathcal{O}(N^3)$ per iteration.
For this case, if the number of RNN stages $L$ is significantly less than $N$, then the overall complexity remains $\mathcal{O}(N^3)$, similar to the generalized WF for interferometric inversion approach in~\cite{yonel2019generalization}.
On the other hand, for achieving an accuracy level of $\epsilon_{WF}\in\mathbbm{R}^+$, the computational cost of the WF approach is $\mathcal{O}(N^2\log N\log(\frac{1}{\epsilon_{WF}}))$~\cite{candes2015phase_IEEE}.

\section{Numerical Simulations}
\label{sec:numerical_results}
In this section, we demonstrate the feasibility of our DL-based phaseless imaging approach through the training and subsequent performance evaluations on a number of real and simulated datasets, with measurement geometries of both experimental and practical interest.
The main objectives of our numerical simulations are the following:
\begin{enumerate}
\item Demonstrating the reconstruction performance of our approach on both real and synthesized datasets, and comparing with the reconstruction results obtained using the WF algorithm~\cite{candes2015phase_IEEE, yonel2020deterministic} and comparable DL-based state-of-the-art phaseless imaging methods, in order to highlight the relative advantages of our approach over a range of image sets.
\item Numerically verifying the robustness of our approach under additive noise on the intensity measurements for relatively low $\frac{M}{N}$ values.
\item Numerically verifying a number of theoretical observations and insights presented in Section~\ref{sec:convergence}.
These include showing the improved accuracy of the initial encoded image, resulting from the inclusion of $\calG$, compared to the accuracy of the spectral estimation, observing the sample complexity improvement compared to the WF algorithm~\cite{candes2015phase_IEEE, yonel2020deterministic} as well as other DL based approaches, and observing the necessity of having ample training set sizes for $\calH$ to appropriately model various image classes of interest.
\end{enumerate}
We adopt the normalized mean squared error (MSE) as the figure of merit throughout this section, and it is defined as $\mathrm{MSE} = \frac{1}{T_s}\sum_{t = 1}^{T_s}\|\hat{\brho}_t - \brho^*_t\|^2 / \|\brho^*_t\|^2$.
$T_s$ is the number of samples in the test set, $\bbD_{test}$, and $\hat{\brho}_t$ and $\brho^*_t$ denote the reconstructed and the corresponding ground truth images, respectively, for the $t^{th}$ sample of $\bbD_{test}$.

\subsection{Dataset Descriptions}
\begin{figure}[!t]
\centering
\includegraphics[width=0.75\columnwidth]{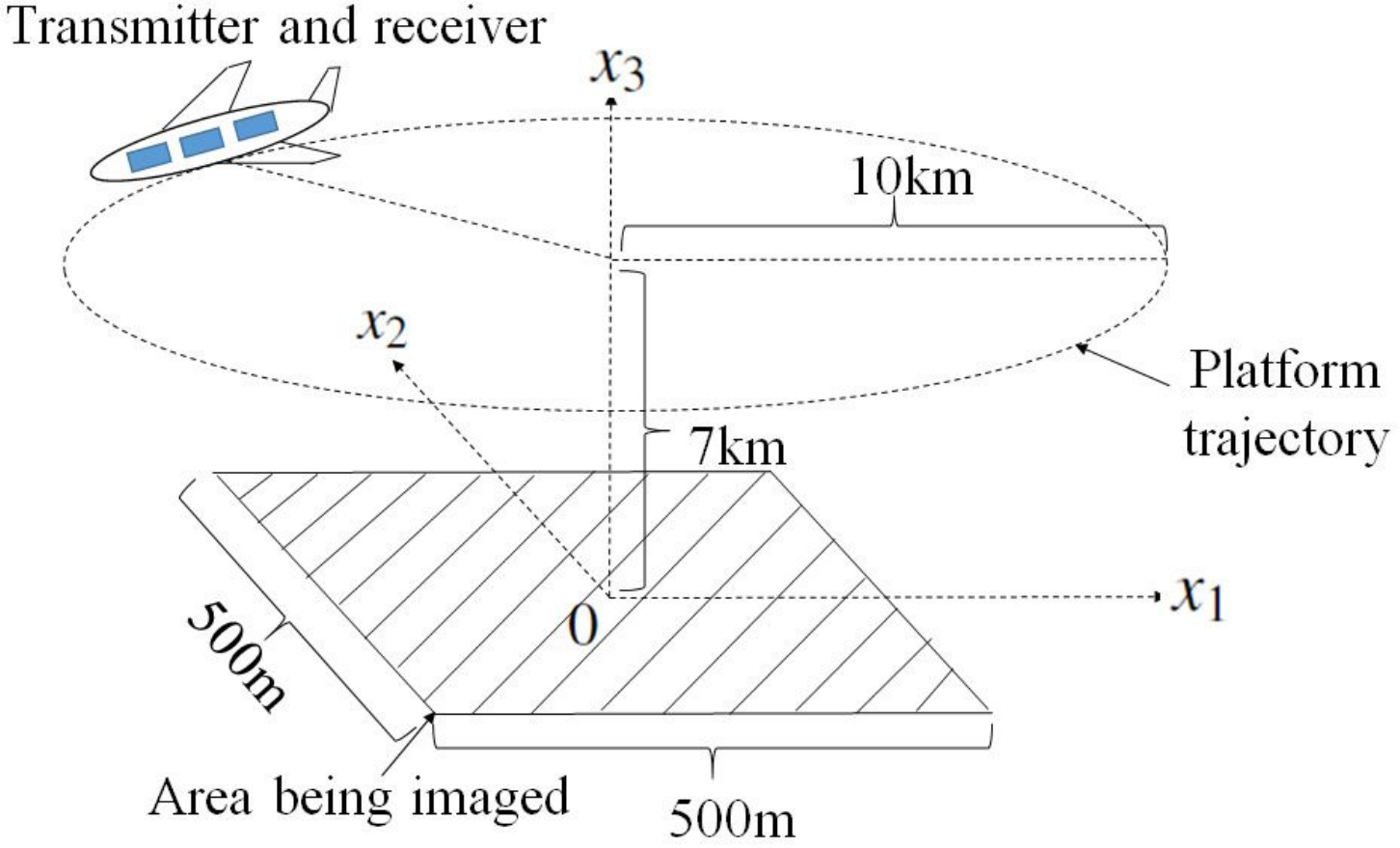}%
\caption{Data collection geometry for the synthetic aperture imaging.}
\label{fig:synthetic_aperture_geometry}
\end{figure}
\begin{figure}[!t]
\centering
\includegraphics[width=0.75\columnwidth]{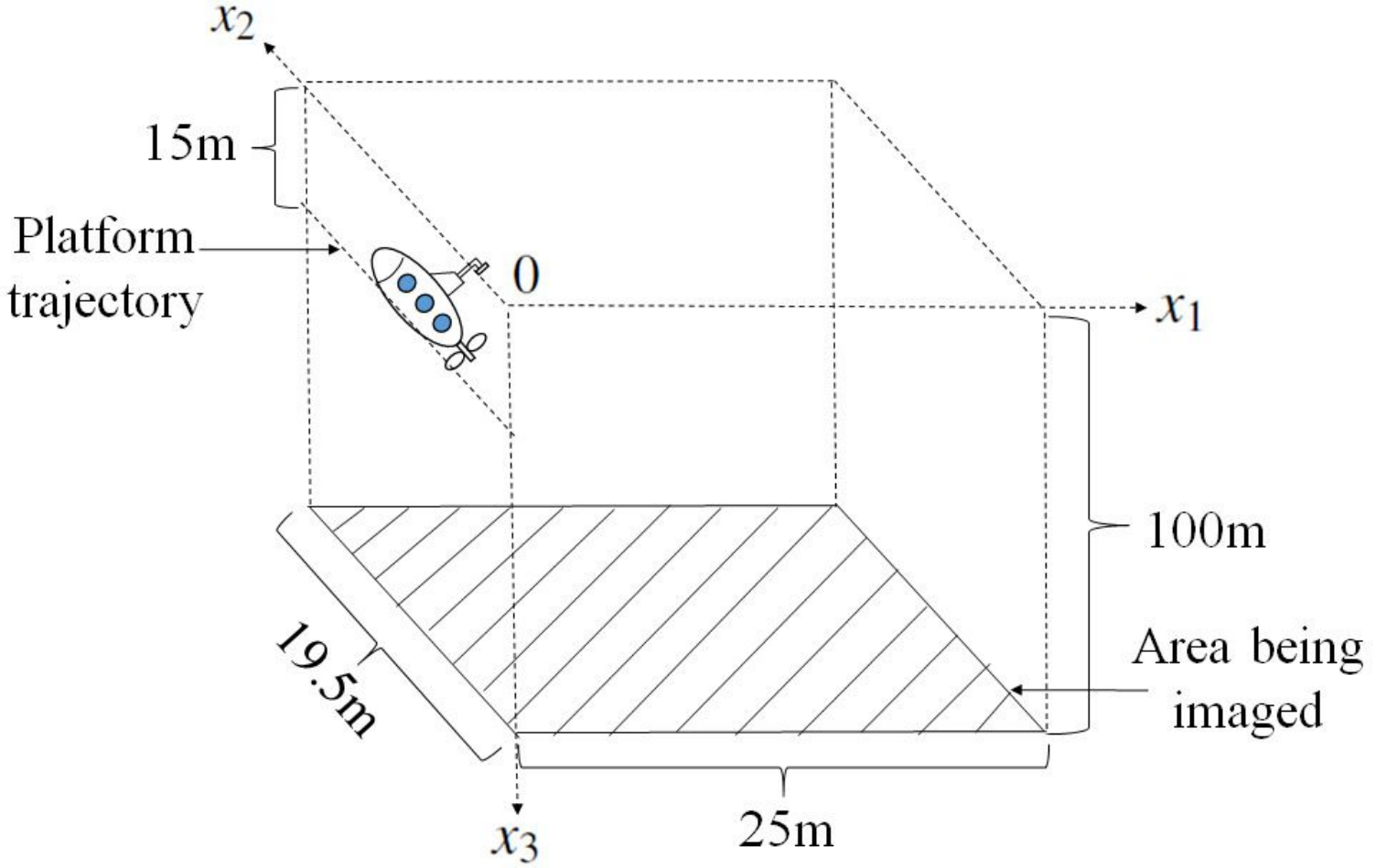}%
\caption{Data collection geometry for PCSWAT dataset for SAS imaging.}
\label{fig:sonar_geometry}
\end{figure}
\begin{figure*}[!t]
\centering
\includegraphics[width=1.4\columnwidth]{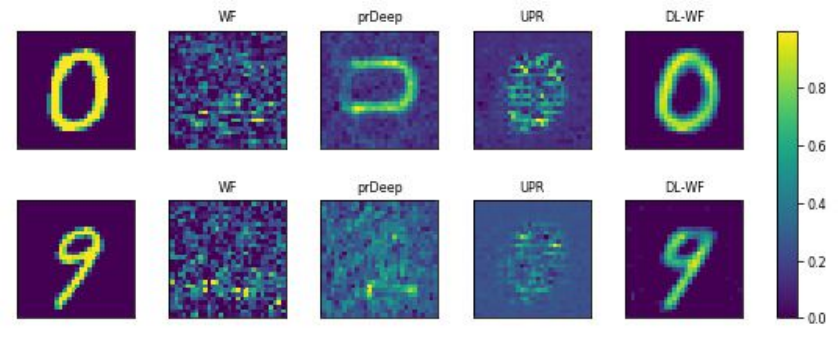}%
\caption{First column includes the original unknown images of dimension $28 \times 28$ pixels. For $M = 0.5N$ and $10000$ training samples, the reconstructed images using the WF algorithm~\cite{candes2015phase_IEEE} with $5000$ iterations are shown in the second column.
Corresponding estimated images using the prDeep~\cite{Metzler2018_prdeep} and the UPR~\cite{naimipour2020_upr} approaches are included in the third and fourth columns, respectively.
The last column shows the estimated images using our method with $10$ RNN stages.}
\label{fig:examples}
\end{figure*}
%
\begin{figure}[!t]
\centering
\includegraphics[width=1.0\columnwidth]{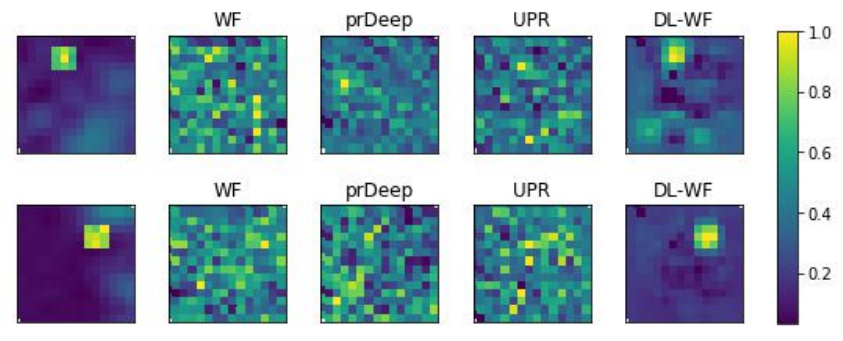}%
\caption{Image reconstruction results for the simulated synthetic aperture dataset with $14\times 14$ pixel images and $\mathrm{SNR} = 10$dB. For $M = N$ and $9950$ training samples, the five columns show the original images, and the reconstructed images using the WF algorithm~\cite{candes2015phase_IEEE} with $5000$ iterations, prDeep approach~\cite{Metzler2018_prdeep}, UPR approach~\cite{naimipour2020_upr} and our DL-WF approach with $10$ RNN stages, respectively.}
\label{fig:mine_example}
\end{figure}
%
\begin{figure}[!t]
\centering
\includegraphics[width=1.0\columnwidth]{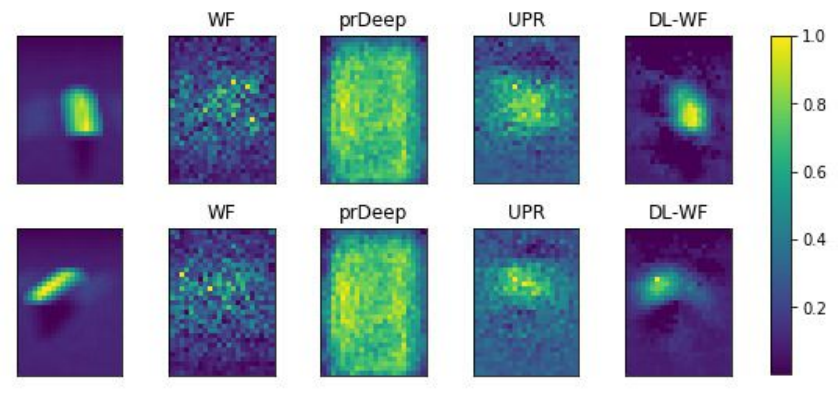}%
\caption{Image reconstruction results for the PCSWAT dataset with $22\times 31$ pixel underwater scenes and SAS measurements. For $M = 930$ and $800$ training samples, the five columns show the original images, and the estimated images using the WF algorithm~\cite{candes2015phase_IEEE} with $5000$ iterations, prDeep approach~\cite{Metzler2018_prdeep}, UPR approach~\cite{naimipour2020_upr} and our DL-based approach with $10$ RNN stages, respectively.}
\label{fig:pcswat_example}
\end{figure}
\begin{figure}[!t]
\centering
\includegraphics[width=0.7\columnwidth]{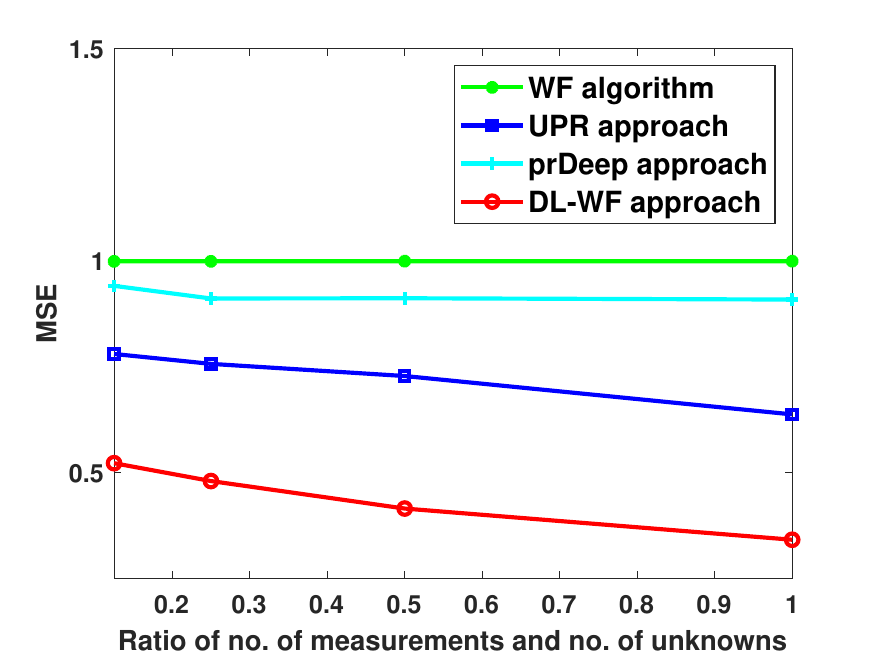}%
\caption{MSE versus $\frac{M}{N}$ for the samples in the MNIST test set.}
\label{fig:err_vs_sample_complexity}
\end{figure}
\begin{figure*}[!t]
\centering
\subfloat[]{\includegraphics[width=0.65\columnwidth]{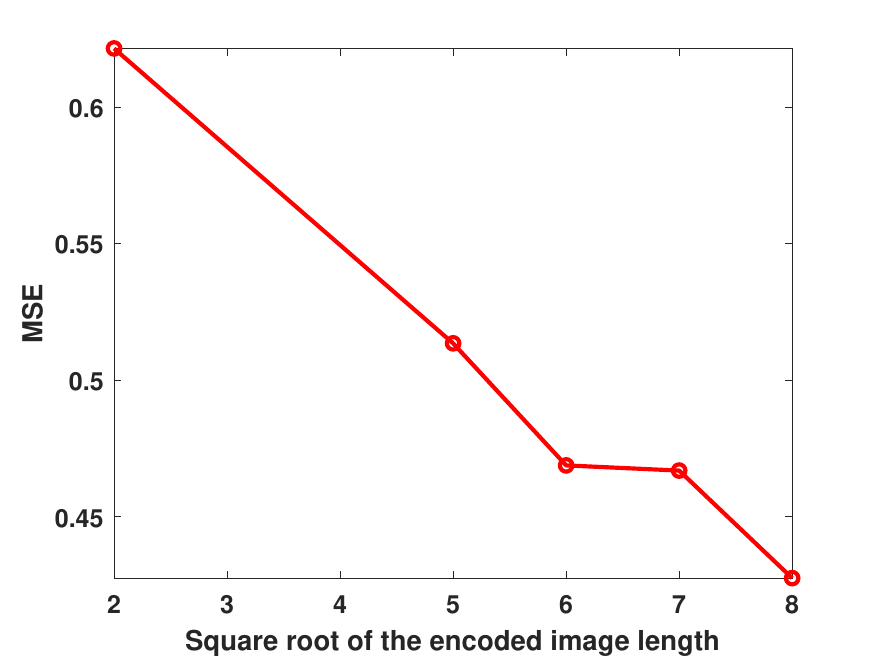}%
\label{fig_mse_vs_Nyequal}}
\hfil
\subfloat[]{\includegraphics[width=0.65\columnwidth]{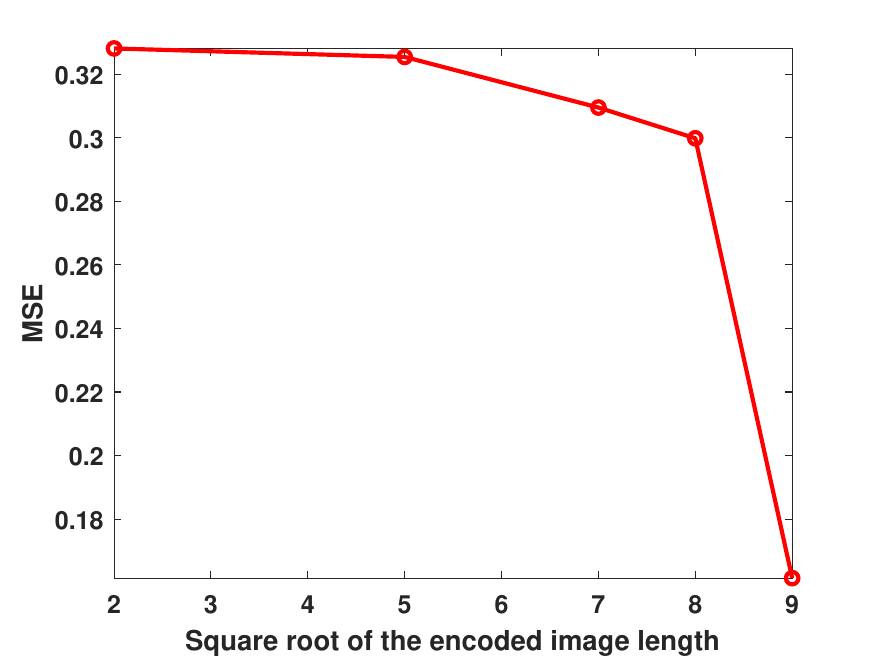}%
\label{fig_mse_vs_Nyequal_mine}}
\hfil
\subfloat[]{\includegraphics[width=0.65\columnwidth]{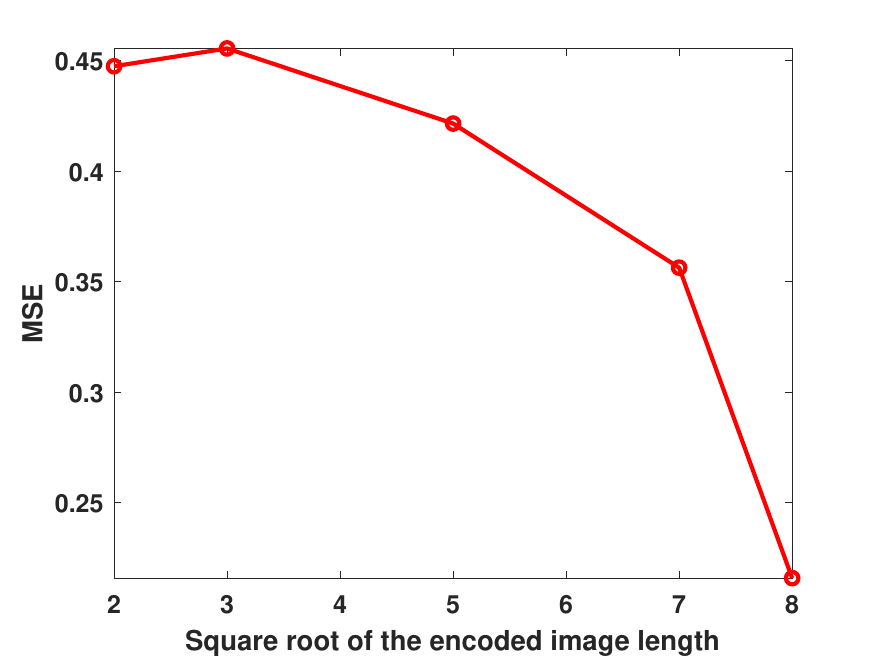}%
\label{fig_mse_vs_Nyequal_pcswat}}
\hfil
\caption{MSE versus $\sqrt{N_y}$ for fixed $M / N$ ratio for the samples in the test set (a) for the MNIST dataset with $M = 0.5N$, (b) for the synthetic aperture dataset with $M = N$, and (c) for the PCSWAT dataset with $M = 1.36N$, and the number of RNN stages $L = 5$.}
\label{fig:mse_vs_Ny_equal}
\end{figure*}
\begin{figure}[!t]
\centering
\includegraphics[width=0.7\columnwidth]{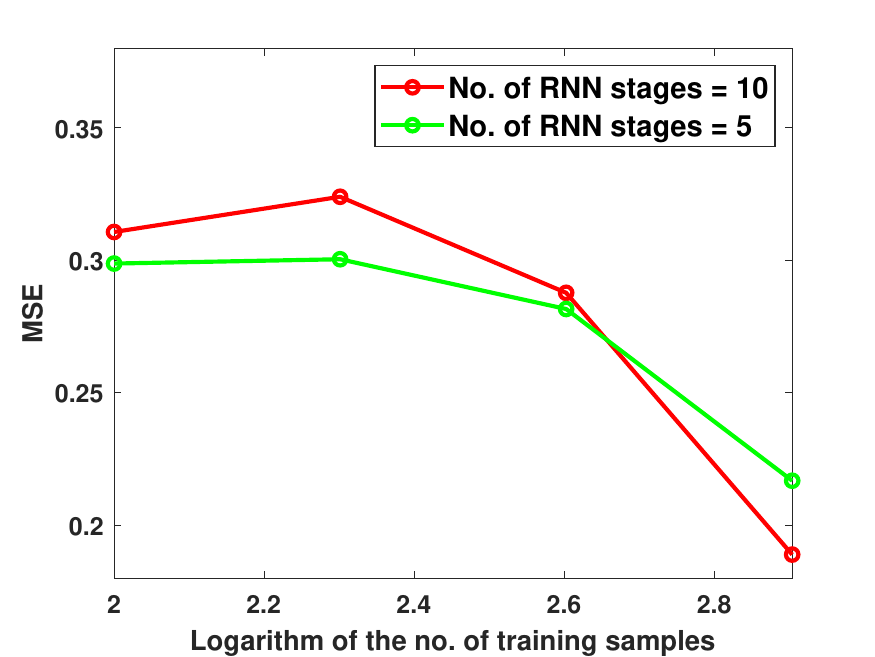}%
\caption{MSE versus the training set sizes for the PCSWAT dataset with $M = 1.36N$.}
\label{fig:mse_vs_training_size}
\end{figure}
\begin{figure}[ht]
\centering
\includegraphics[width=0.7\columnwidth]{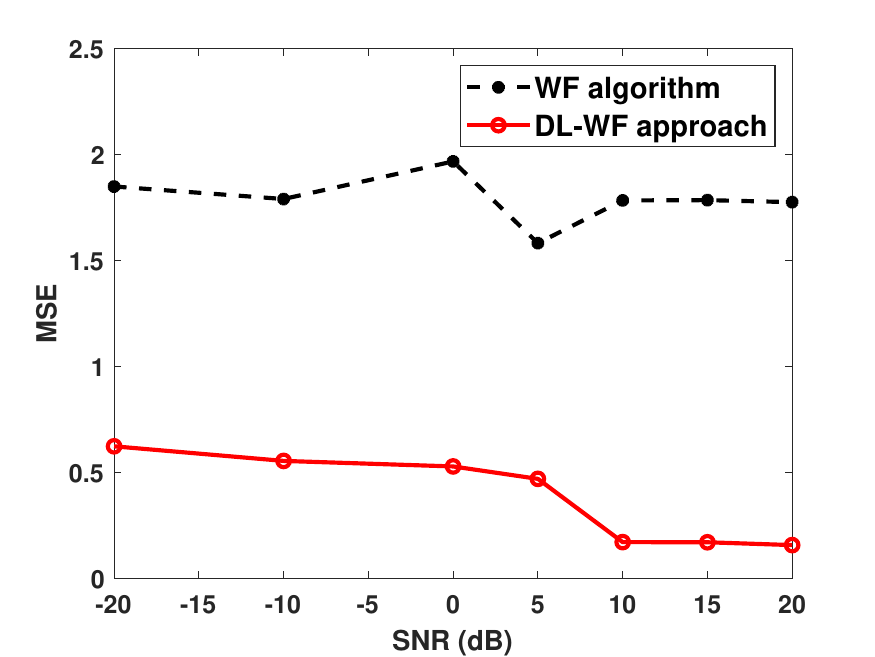}
\caption{MSE versus SNR (dB) for $M = N$ for the samples in the synthetic aperture test dataset.}
\label{fig:mse_vs_SNR_mine}
\end{figure}
In this subsection, we introduce three image sets, and the associated deterministic forward maps that results from three different data acquisition geometries.

\subsubsection{MNIST Dataset}
The first image set that we consider in this paper is MNIST, which is a publicly available dataset of handwritten digits.
Each image has a dimension of $28 \times 28$ pixels, and depicts one of the $10$ digits.
We randomly select $10000$ samples, with $1000$ samples for each digit, as the training dataset, and another randomly selected $100$ images, $10$ for each digit, constitute the test set.
For the forward mapping matrix, we use the one available with the publicly available dataset from~\cite{metzler2017coherent} for the $40 \times 40$ pixels imaging scenario.
This dataset considers a multiple scattering transmission environment with phaseless measurements, and the forward map is recovered using the prVAMP based double phase retrieval approach.
Since our images have a lower pixel count, we consider the first $784$ columns of this matrix to form our forward map $\bF$, and discard the phases of the $\bF\brho^*_t$ values to form the phaseless measurements for the images in the MNIST dataset.
The number of rows of $\bF$, which is the number of total measurements $M$, is varied for experimentation purposes, and for each case, we consider the first $M$ rows of $\bF$.

\subsubsection{Simulated Synthetic Aperture Dataset}
The second dataset is selected with the goal of showing a scenario where our approach is applicable in a practical setting with a deterministic forward map.
We apply our method for synthetic aperture imaging~\cite{jao2001theory} from simulated measurement under Born approximation.
Each scene being imaged has a dimension of $500$m$\times 500$m and is reconstructed as a $14 \times 14$ pixels image.
There is a single square object located at a random location within the area being imaged, and the background varies from scene to scene.
The number of samples in the training and test sets are $9950$ and $50$, respectively.
We consider a mono-static data-collection strategy, with the transmitter-receiver trajectory along a circular path at $7$km height and at a radius of $10$km.
Total number of measurements is set equal to the number of unknowns, i.e., $M = 196$, and additive Gaussian noise of zero mean and different variances is assumed to be present in the measured intensity values of the received signal.
A schematic diagram of the associated data collection geometry is shown in Fig.~\ref{fig:synthetic_aperture_geometry}.

\subsubsection{PCSWAT Generated SAS Dataset}
For the third dataset, we consider a PCSWAT 10 software generated simulated dataset for synthetic aperture sonar (SAS) imaging of underwater scenes.
PCSWAT is a tool-set developed for simulating high-fidelity SAS data~\cite{marx2000introduction}.
It offers a selection of realistic targets and underwater surface types, and allows the incorporation of varying sound-velocity profiles, marine life property, wind speed etc.
For the samples in our training and test sets, we consider that the background medium is composed of sandy gravel, and there is sparse marine life present in the medium.
Each scene contains a single hemi-spherically end-capped cylinder of varying length and fixed radius located at a random location on the scene along a random orientation.
Each area being imaged has a dimension of $19$m$\times 25$m and it is reconstructed as a $22 \times 31$ pixels image.
The number of samples in the training and test set are $800$ and $5$, respectively.
We consider a 2D environment, and the vehicle and the water depth are set to $15$m and $100$m, respectively.
The center frequency and the bandwidth of the transducers mounted on the moving vehicle are set equal to $120$kHz and $30$kHz, respectively.
The data collection geometry for the SAS operation simulated via PCSWAT is shown in Fig.~\ref{fig:sonar_geometry}.

\subsection{DL Architectures and Reconstruction Results}
\label{subsec:architecture_reconstruction}
The quality of the reconstructed images is heavily dependent on the $\calG$ and $\calH$ network architectures.
For evaluating our numerical results for the MNIST dataset, we consider the following network model: the number of RNN stages is set to $10$; for $\calG$, we use a $5$ layer CNN model with $leaky\_relu(.)$ activation functions from the tensorflow library, and output dimensions $24 \times 24 \times 4$, $20 \times 20 \times 16$, $16 \times 16 \times 16$, $12 \times 12 \times 4$ and $8 \times 8 \times 1$ for the $5$ consecutive layers; for $\calH$, we have used a $5$ layer ANN architecture with $relu(.)$ activation functions, and output vector lengths of $64$, $64$, $64$, $100$ and $784$, respectively.
Additionally, since the maximum value of each MNIST image is $255$, we consider this as prior information while applying our DL-based approach as well as the other methods that we evaluate for comparison.
This is performed by first normalizing the set of intensity vectors by $255^2$, and then adding the term $\frac{1}{T}\sum_{t = 1}^T\left(\max_{n\in[N]}\hat{\brho}_t(n) - 1\right)^2$, where $\hat{\brho}_t(n)$ denotes the $n^{th}$ element of $\hat{\brho}_t$, to the training loss functions of the end-to-end DL-based methods.
On the other hand, for the iterative methods including the WF algorithm, we normalize the updated image estimation at each iteration so that the maximum pixel value equals to $1$.

The network models implemented for the synthetic aperture and the PCSWAT generated SAS datasets are kept similar as the ones used for MNIST.
For synthetic aperture imaging, the number of filters in $\calG$ network layers are $8$, $12$, $12$, $8$ and $1$, respectively, while the output vector lengths of the $5$ consecutive layers of $\calH$ are $81$, $85$, $90$, $100$ and $196$, respectively.
For the SAS dataset, the number of filters used in the $\calG$ network layers are the same, except, in this case, the encoded image dimension $N_y$ is set to $64$.
The output vector lengths for the $5$ consecutive $\calH$ layers are $64$, $81$, $100$, $150$ and $200$, respectively.
We used ADAM optimizer for training, with a learning rate equal to $10^{-5}$, and batch sizes equal to $100$, $50$ and $5$ for the MNIST, synthetic aperture and the PCSWAT datasets, respectively.

For the remainder of this section, we refer to our approach by DL-WF.
We note that while comparing with other DL-based state-of-the-art phaseless imaging approaches, we exclude comparisons with the generative prior based methods~\cite{hand2018phase, shamshad2018robust, hand2020compressive, shamshad2020robust}, as they require us to separately train a GAN using a large amount of images from comparable image classes.
One of the motivations of our approach is to avoid the cumbersome GAN network training, and instead adopt an end-to-end training strategy that uses sample images and the corresponding intensity measurements.
Additionally, although our theoretical results do not guarantee performance improvement over the generative prior based methods, there are several advantages of our exact recovery guarantee compared to the theoretical results derived in~\cite{hand2018phase, hand2020compressive} as summarized in Section~\ref{sec:convergence}.
In this section, we instead include reconstruction results from two state-of-the-art DL-based approaches with comparable training complexities, namely, UPR~\cite{naimipour2020_upr} and prDeep~\cite{Metzler2018_prdeep}.
UPR~\cite{naimipour2020_upr} uses an end-to-end training scheme, with similar training dataset requirement as DL-WF.
On the other hand, prDeep is a regularization by denoising~\cite{romano2017little} type approach for phaseless imaging, and it implements a DnCNN~\cite{zhang2017_dncnn} for denoising.
We have used a $17$ layer DnCNN network with a similar architecture as the one presented in~\cite{zhang2017_dncnn}, where the number of channels at each intermediate layer is $64$.
Instead of patch extraction, due to the relatively small image dimensions under consideration here, we apply the entire image as input to the denoising network.
For additional implementation details for UPR and prDeep, we followed the various hyper-parameter values suggested in~\cite{naimipour2020_upr} and~\cite{Metzler2018_prdeep}, respectively.

Example reconstructed images using our DL-based method along with the reconstructed images using the WF algorithm, prDeep and UPR are shown in Fig.~\ref{fig:examples},~\ref{fig:mine_example} and~\ref{fig:pcswat_example} for the MNIST, synthetic aperture and PCSWAT datasets, respectively.
The number of training samples used for these three datasets are $10000$, $9950$ and $800$, respectively.
For all three cases, we observe that our DL-WF approach yields significant improvement in the estimated image accuracies compared to the WF algorithm, and the DL-based UPR and prDeep methods.

Despite this improvement, the reconstructed images produced by our approach still retain visual artifacts. There are two key contributors to this end. Firstly, Section~\ref{sec:recovery_guarantees} discusses exact recovery with linear convergence for elements representable in the range of $\calH$.
For the particular training dataset and the optimization algorithm used for training, whether we can estimate an $\calH$ with the properties specified in the theory of exact recovery is an additional aspect that contributes to empirically observing such guarantees. Secondly, even under the validity of these assumptions, observing exact recovery still potentially requires many iterations of gradient updates in the lower dimensional encoded space given the ill-posed problem settings under consideration. The architecture is however limited by the number of layers in the RNN unit, hence convergence to the true solution is not necessarily observed. Accordingly in Fig.~\ref{fig_mse_vs_no_of_RNN_stages}, we demonstrate the expected decaying trend in average reconstruction error as the number of RNN stages are increased. Furthermore, despite these limitations, the improvements in the reconstruction quality that our approach offers over the state-of-the-art phaseless imaging methods is still quite significant.

\vspace{-0.1in}
\subsection{Effect of Sample Complexity}
In order to observe the effect of sample complexity on our approach, MSE values for the MNIST test dataset are plotted versus the $\frac{M}{N}$ ratios in Fig.~\ref{fig:err_vs_sample_complexity}.
It is observed that, for each of the $\frac{M}{N}$ ratios, our DL-based approach performs better compared to the WF algorithm, prDeep and UPR.
Additionally, as expected intuitively, we observe reduced MSE values as $M$ is increased for a fixed image dimension.
In Fig.~\ref{fig_mse_vs_Nyequal},~\ref{fig_mse_vs_Nyequal_mine} and~\ref{fig_mse_vs_Nyequal_pcswat}, we consider the case, where the number of measurements $M$ is fixed while the encoded image dimension $N_y$ is varied, and we plot the MSE values versus $\sqrt{N_y}$ for the MNIST, synthetic aperture and the PCSWAT datasets, respectively.
For all three cases, we observe reduced MSE values with increasing $N_y$.
Our observation from Fig.~\ref{fig:mse_vs_Ny_equal} implies that the reconstruction in the encoded image space $\bbY$, reveals a latent dimension of the images smaller than the number of unknowns.
This indicates that, compared to the WF algorithm and the state-of-the-art DL-based methods, our approach has lower sample complexity requirement since it searches over a reduced space for the unknown image, as supported by the corresponding MSE values in Fig.~\ref{fig:err_vs_sample_complexity}.

\subsection{Effect of the Number of Training Samples}
Another important criteria is the necessity of having an adequate number of training samples for effective image reconstruction at the decoding network output.
In Fig.~\ref{fig:mse_vs_training_size}, we plot the MSE values for the test set versus the number of training set sizes for $M = 1.36N$ for the PCSWAT dataset.
We consider two cases with the same $\calG$ and $\calH$ network architectures, and the number of RNN stages equal to $5$ and $10$.
As expected, we observe for both cases that an increasing training set size helps $\calH$ to capture the underlying image prior more effectively as long as the $\calH$ architecture has sufficient capacity.
Additionally, it helps the encoder to learn a better encoded image space while simultaneously attaining improved RNN parameter values, which translates into lower MSE values at similar stages of the training process.
As for the two curves corresponding to the different numbers of RNN layers, we observe that as long as the overall imaging network is sufficiently trained, which is represented by the last points on both curves, increasing the number of RNN layers helps improve the reconstruction accuracies.

\vspace{-0.1in}
\subsection{Accuracy of the Initial Value}
In order to observe the effect of $\calG$ on the initialization accuracy and to indirectly verify the observation from~\eqref{eq:dist_y}, we consider three separate mean initialization error related terms for the samples in the test sets, namely, $d_1$, $d_2$ and $d_3$.
We define $d_1$, $d_2$ and $d_3$ as $d_1 = \frac{1}{T_s}\sum_{t = 1}^{T_s}\|\brho^{(0)}_t - \brho^*_t\|^2 / \|\brho^*_t\|^2$, $d_2 = \frac{1}{T_s}\sum_{t = 1}^{T_s}\|\calH(\by^{(0)}_t) - \brho^*_t\|^2 / \|\brho^*_t\|^2$ and $d_3 = \frac{1}{T_s}\sum_{t = 1}^{T_s}\|\by^{(0)}_t - \by^{(L)}_t\|^2 / \|\by^{(L)}_t\|^2$.
A more accurate calculation of the initialization error for the encoded image space, $d_3$, requires $\by^{(L)}_t$ to be replaced by $\by^*_t$.
When the $\calG$ and $\calH$ network architectures, and the numbers of training samples are set as described in Subsection~\ref{subsec:architecture_reconstruction}, and the number of RNN stages is set to $10$, we observe that for the three datasets, the three initialization error related terms have the following values: for MNIST with $M = 0.5N$, $d_1 = 209.344$, $d_2 = 0.999989$ and $d_3 = 0.000145729$; for the synthetic aperture dataset with $M = N$, $d_1 = 2.02657$, $d_2 = 0.300721$ and $d_3 = 0.00103663$; and for the PCSWAT dataset with $M = 1.36N$, $d_1 = 1.49407$, $d_2 = 0.525142$ and $d_3 = 0.00129634$.
In all three cases, we observe that the trained networks produce significantly reduced initialization errors for the encoded image space compared to the ones for the original image space.

\vspace{-0.1in}
\subsection{Effect of SNR of the Intensity Measurements}
The effect of varying SNR values, resulting from the different levels of noise detected at the receiving sensors along with the intensity values of the reflected signals, on the accuracies of the reconstructed images for the synthetic aperture dataset is shown in Fig.~\ref{fig:mse_vs_SNR_mine}.
We compare these values to the corresponding image reconstruction accuracies for the WF algorithm.
With increasing SNR, we observe some reduction in the MSE values, calculated after a fixed number of training updates.
For each case, our DL method is observed to significantly outperform the WF algorithm.

\subsection{Effect of $\calG$ and $\calH$ Architectures and No. of RNN Layers}
\begin{figure*}[!t]
\centering
\subfloat[]{\includegraphics[width=0.68\columnwidth]{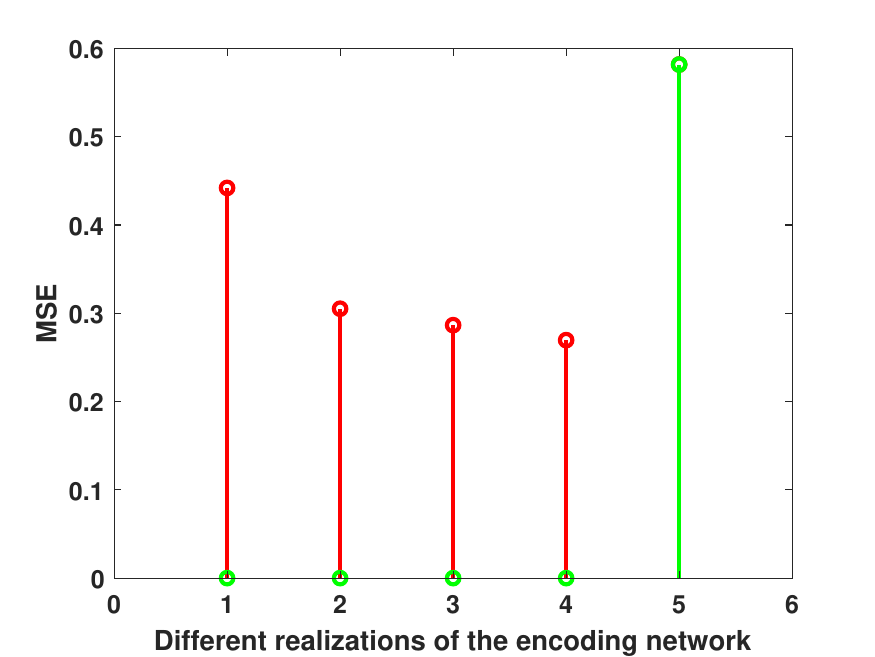}%
\label{fig_capacity_G_RNN5}}
\subfloat[]{\includegraphics[width=0.68\columnwidth]{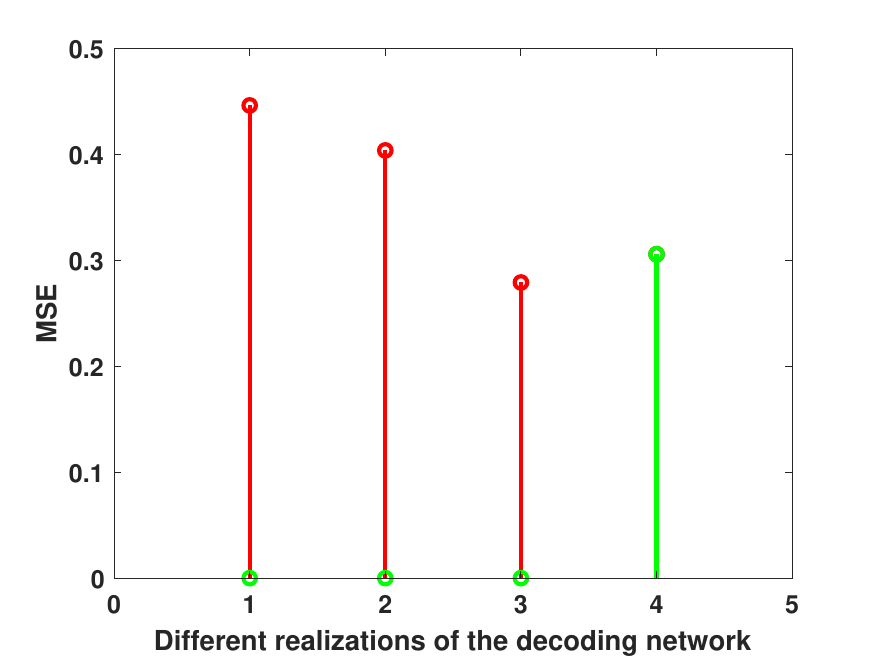}%
\label{fig_capacity_H_RNN5}}
\subfloat[]{\includegraphics[width=0.68\columnwidth]{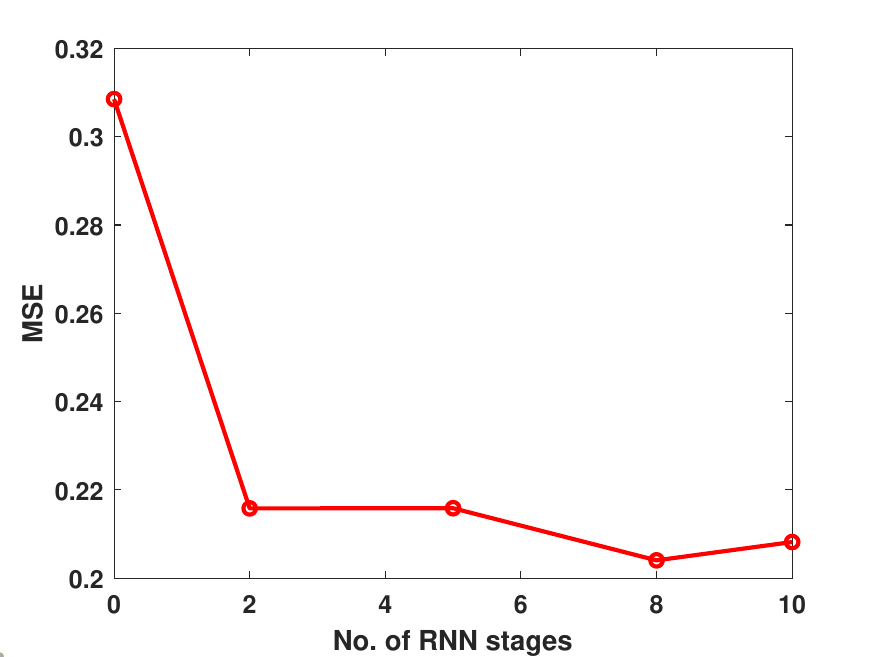}%
\label{fig_mse_vs_no_of_RNN_stages}}
\caption{MSE for the PCSWAT dataset using different (a) $\calG$ and (b) $\calH$ network architectures and $5$ RNN layers. The last points on the x-axis correspond to the linear encoding and decoding network architectures in (a) and (b), respectively. (c) shows the MSE values obtained by using different numbers of RNN stages for fixed $\calG$ and $\calH$ network architectures.}
\label{fig:capacity_G_H}
\end{figure*}
In this subsection, we show the effects of the encoding and decoding network architectures, and the number of RNN stages on the performance of our approach.
We consider the PCSWAT dataset for this purpose, and while evaluating the effect of each of these criteria, for example the $\calG$ network architecture, we keep the remaining elements, i.e. the $\calH$ network architecture and the number of RNN layers, unchanged.
In Fig.~\ref{fig_capacity_G_RNN5} and~\ref{fig_capacity_H_RNN5}, each point along the $x$-axis, corresponds to one realization of the $\calG$ and $\calH$ networks, respectively.
The number of parameters for these different architectures increase from left to right, and the linear architectures for each case are indicated by the last points, where the corresponding linear networks have the maximum number of trainable parameters.
These figures provide the empirical observation that, the MSE value reduces with $\calG$ and $\calH$ network architectures with increasing number of parameters, and a linear encoder is more detrimental than a linear decoder.
Finally, Fig.~\ref{fig_mse_vs_no_of_RNN_stages} verifies that with an increasing number of RNN layers, we can improve our reconstruction quality.

\subsection{Comparison with DL Methods for Fourier Measurements}
\begin{figure}[!t]
\centering
\includegraphics[width=1.0\columnwidth]{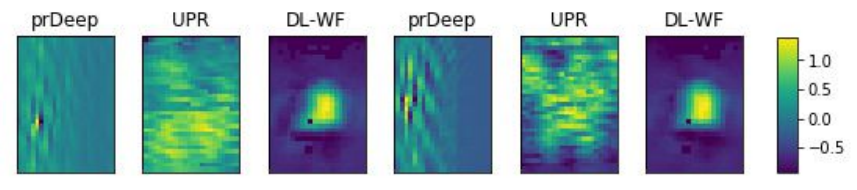}%
\caption{Reconstructed images for the samples in the PCSWAT test imageset from the corresponding Fourier measurements using different DL-based methods. For the ground truth image in the top left corner of Fig.~\ref{fig:pcswat_example}, the first three figures show the reconstructed images using prDeep, UPR and DL-WF, respectively, for $M = 1.5N$; next three figures show the corresponding images for $M = 2N$.}
\label{fig:pcswat_example_fourier}
\end{figure}
In this subsection, we compared our approach to prDeep and UPR, under Fourier measurement models.
We use the images from the PCSWAT dataset along with two cases of the Fourier forward map, where the number of measurements equal to $1.5$ times and $2$ times the number of unknowns.
For our DL-based approach, we adopt a $5$ layer RNN, with $\calG$ and $\calH$ network architectures as presented in Subsection~\ref{subsec:architecture_reconstruction} for the PCSWAT dataset.
Example reconstructed images resulting from the three approaches are shown in Fig.~\ref{fig:pcswat_example_fourier}.
For both values of $\frac{M}{N}$, we observe that our approach outperforms the state-of-the-art DL-based methods under Fourier measurements.

\subsection{Comparison of Computation Time}
Computation times during the inversion phases for our approach and the UPR method are dominated by the time required to compute the spectral initialization output.
On the other hand, for the prDeep approach, significant computational time is involved both for training the denoising network and the inversion phase.
Over all three datasets, we observe that the $20$-layer RNN network for the UPR approach has the lowest computational time, followed by our DL-WF approach, while the WF algorithm and the prDeep approach require the highest computational time during the inversion phase.
As an example, for the PCSWAT dataset, the average computational times required in the inversion phases of the WF algorithm, prDeep, UPR and DL-WF with $L = 5$ are $3.8815$, $9.8472$, $0.0147$, and $0.0148$ minutes, respectively.
For the same dataset, the training time for the UPR and the DL-WF approaches is approximately $2$-$3$ days, while the training time for the denoising network of the prDeep approach is approximately a few hours.

\section{Conclusion}
\label{sec:conclusions}
In this paper, we introduced a DL-based phaseless imaging approach that incorporates an RNN with DL-based encoding-decoding stages, and determined sufficient conditions for exact recovery guarantee.
Our theoretical results show that, depending on the Lipschitz constants of the encoding and the decoding networks, it is possible to achieve improved convergence rate as compared to the WF algorithm~\cite{candes2015phase_IEEE}.
Additionally, the valid range of forward maps for which the exact recovery guarantee holds is less restrictive than those sufficient conditions introduced in earlier works~\cite{hand2018phase, hand2020compressive, yonel2020deterministic}.
Desired spectral property of the decoder for the feasibility of our recovery guarantee reveals that our theoretical results are consistent with the observations for the case with linear Gaussian generative priors and forward maps with i.i.d. Gaussian distributed elements.
Our numerical simulations show the advantages of our approach, under low sample complexity regimes and deterministic forward maps, over the WF algorithm as well as the existing DL-based methods.
In future work, we will consider extensions to take into account partially known forward maps which relates to a multiple scattering within extended objects scenario in practical remote sensing applications.
Additionally, we will pursue improvements to our approach with deep equilibrium architectures \cite{gilton2021deep} to facilitate more iterations on the lower dimensional encoded space towards higher accuracy in reconstructions.

\appendices
\section{Proof of the Theoretical Results}
\label{sec:proofs}
For the purpose of conciseness, we denote $\brho\brho^H$ and $\brho^*{\brho^*}^H$ by $\tilde{\bP}$ and $\tilde{\bP}^*$, respectively.
Similarly, we use $\tilde{\bY}$ and $\tilde{\bY}^*$ to represent $\by\by^H$ and $\by^*{\by^*}^H$, respectively.
We denote the error term $\calH(\by) - \brho^*$ by $\be_{\brho}$, and the error $\by - \by^*$ by $\be_{\by}$.
For a particular $\by$ in the definitions of $\be_{\by}$ and $\be_{\brho}$, the corresponding $\by^*$ refers to the encoded representation of the element in the solution set $\bbP$ that is closest to $\by$.
Additionally, we denote the multiplication of the Lipschitz constants $\mu_{\calG}$, $\mu_{\calH}$ and $\mu_{\calR}$ by $\mu_M$, and the set composed of these three constants by $\bbV$.

We start by stating important results from~\cite{yonel2020deterministic} that directly translates for the range restricted condition in~\eqref{eq:cond_phaseless2}
from
the universal condition used in~\cite{yonel2020deterministic}.
Specifically, the results from Lemma III.1 and Lemma III.2 from~\cite{yonel2020deterministic} modified for our sufficient  condition in~\eqref{eq:cond_phaseless2} are restated below:
\begin{lemma}
\label{lemma:lemma_III_1}
(Lemma III.1~\cite{yonel2020deterministic}) \\
If the condition in~\eqref{eq:cond_phaseless2} is satisfied for all $\brho^*\in\bbT$, then the operator $\frac{1}{M}\calF^H\calF : \bbZ \mapsto \bbC^{N \times N}$, where $\bbZ = \{\brho^*{\brho^*}^H : \; \forall\brho^*\in\bbT\}$, can be written as
\begin{align}
    \label{eq:FHF_IRD} \frac{1}{M}\calF^H\calF = \mathcal{I} + \calQ + \Delta,
\end{align}
where $\mathcal{I}(\tilde{\bP}^*) = \tilde{\bP}^*$, $\calQ(\tilde{\bP}^*) = \|\brho^*\|^2\bI$, and for all $\brho_1, \brho_2 \in\bbT$,
\begin{align}
    \label{eq:calR_rho1_rho2} \calQ(\tilde{\bP}_1 - \tilde{\bP}_2) = (\|\brho_1\|^2 - \|\brho_2\|^2)\bI.
\end{align}
For $\Delta : \bbZ \mapsto \bbC^{N \times N}$, $\|\Delta(\tilde{\bP}^*)\| \leq \delta\|\brho^*\|^2$ with high probability.
\end{lemma}
\begin{proof}
See proof of Lemma III.1 in~\cite{yonel2020deterministic}.
\end{proof}
\begin{lemma}
(Lemma III.2~\cite{yonel2020deterministic})
Let the condition in~\eqref{eq:cond_phaseless2} is satisfied for all $\brho^*\in\bbT$.
If $\brho^{(0)}$ is set equal to $a_{\rho_0}\bz_0$, where $a_{\rho_0} = \frac{1}{(2M)^{1 / 4}}\sqrt{\|\bd\|}$ and $\bz_0$ is the leading eigenvector of the spectral matrix $\bY$ defined in
\eqref{eq:spectral_matrix},
then
\begin{align}
    \label{eq:dist_rho0_rhoc4} \rmdist(\brho^{(0)}, \brho^*) \leq \epsilon\|\brho^*\|,
\end{align}
where $\epsilon$ is defined in (21) of~\cite{yonel2020deterministic}.
\end{lemma}
\begin{proof}
See proof of Lemma III.2 in~\cite{yonel2020deterministic}.
\end{proof}
For the expression of $\epsilon$ from~\cite{yonel2020deterministic} to be valid, $\delta$ in~\eqref{eq:cond_phaseless2} should be less than $0.5$.

\subsection{Eigenvectors and Eigenvalues of $\bE_{\brho}$}
\label{app:eigenvector_lifted_error}
Let $\lambda_i\in\mathbbm{R}$ and $\bv_i\in\bbC^N$ denote the $i^{th}$ eigenvalue and the corresponding eigenvector of $\bE_{\brho}$, defined in~\eqref{eq:E_rho_def}, for $i\in[2]$.
Since $\bE_{\brho}\bv_i$ can be expanded as $\be_{\brho}(\calH(\by)^H\bv_i) + \brho^*(\be^H_{\brho}\bv_i)$, by left multiplying both sides of
\begin{align}
    \label{eq:lambda_v} \lambda_i\bv_i & = \be_{\brho}(\calH(\by)^H\bv_i) + \brho^*(\be^H_{\brho}\bv_i),
\end{align}
by $\calH(\by)^H$ followed by a rearrangement step, $\calH(\by)^H\bv_i$ can be expressed as $\alpha_i\be^H_{\brho}\bv_i$ where $\alpha_i := \frac{\calH(\by)^H\brho^*}{\lambda_i - \calH(\by)^H\be_{\brho}}$.
By substituting this expression on the right hand side of~\eqref{eq:lambda_v}, we can express $\bv_i$ as
\begin{align}
    \label{eq:bv_i} \bv_i = \frac{1}{\lambda_i}(\be^H_{\brho}\bv_i)\left[\alpha_i\calH(\by) + (1 - \alpha_i)\brho^*\right],
\end{align}
from which it is evident that $\bv_i$ is a weighted summation of $\calH(\by)$ and $\brho^*$.
Additionally, since~\eqref{eq:bv_i} can be rearranged as
\begin{align}
\lambda_i\bv_i = \left[\alpha_i\calH(\by) + (1 - \alpha_i)\brho^*\right]\be^H_{\brho}\bv_i,
\end{align}
we observe that $\lambda_i$ and $\bv_i$ are respectively the eigenvalue and corresponding eigenvector of a rank-1 matrix $\left[\alpha_i\calH(\by) + (1 - \alpha_i)\brho^*\right]\be^H_{\brho}$.
From this observation, $\lambda_i$ and $\bv_i$ can be directly expressed as $\lambda_i = \be^H_{\brho}\bt_i$ and $\bv_i = \bt_i / \|\bt_i\|$, where $\bt_i = \alpha_i\calH(\by) + (1 - \alpha_i)\brho^*$.
Additionally, by replacing $\lambda_i$ by $\be^H_{\brho}\bt_i$ in the definition of $\alpha_i$ followed by a rearrangement step, we get the following quadratic equation of $\alpha_i$:
\begin{align}
    \|\be_{\brho}\|^2\alpha^2_i + (\be^H_{\brho}\brho^* - \calH(\by)^H\be_{\brho})\alpha_i - \calH(\by)^H\brho^* = 0.
\end{align}
The solutions to this equation provide the required expressions of $\alpha_i$'s as functions of $\calH(\by)$ and $\brho^*$.

\subsection{Derivation of~\eqref{eq:finalCond} from~\eqref{eq:Hcondition}}
\label{app:derivation_26_27}
To observe how~\eqref{eq:Hcondition} delivers the desired bound on the perturbation operator in~\eqref{eq:finalCond}, we use the definition of $\tilde{\calH}$, which extracts the rank-1 PSD decomposition of its input to generate elements in the range of $\calH$.
Under~\eqref{eq:Hcondition},
\begin{align}
    \| \Delta(\tilde{\calH}(\tilde{\bY}) -  \tilde{\calH}(\tilde{\bY}^*)) \| & \leq \omega \|  \Delta ( \calH (\sqrt{\lambda_0} \bu_0 ) \calH (\sqrt{\lambda_0} \bu_0 )^H) \|,
\end{align}
where $\lambda_0$ and $\bu_0$ are the leading eigenvalue and eigenvector of $\tilde{\bY} - \tilde{\bY}^*$, respectively.
Under the assumption that $\bbY$ is an affine set, it is evident from the spectral analysis of $\tilde{\bY} - \tilde{\bY}^*$ that $\bu_0\in\bbY$.
Therefore, by applying our sufficient condition from~\eqref{eq:cond_phaseless2}, we observe that
\begin{align}
    \| \Delta(\tilde{\calH}(\tilde{\bY}) -  \tilde{\calH}(\tilde{\bY}^*) \| \leq \omega\delta \| \calH (\sqrt{\lambda_0} \bu_0 ) \|^2.
\end{align}
Now, using the Lipschitz constant of $\calH$, $\| \calH (\sqrt{\lambda_0} \bu_0 ) \|^2$ is upper bounded by $\mu^2_{\calH}\lambda_0$.
Since $\lambda_0 = \| \tilde{\bY} - \tilde{\bY}^* \| \leq \| \tilde{\bY} - \tilde{\bY}^* \|_F $ by definition, we get inequality relation in~\eqref{eq:finalCond}.

\subsection{Derivation of~\eqref{eq:dist_y} and~\eqref{eq:k_ub_lb}}
\label{app:initialization}
We begin this subsection by defining $b_1(\mu_{\calG}, \mu_{\calH}, \epsilon)$ and $b_2(\bbV, \epsilon)$, that appear in the lower bound expression in~\eqref{eq:k_ub_lb}, as follows:
\begin{align}
    \label{eq:b1_def} b_1(\mu_{\calG}, \mu_{\calH}, \epsilon) & := \frac{1}{\epsilon}\left(\frac{1}{\mu_{\calH}} - \mu_{\calG}(1 + \epsilon)\right), \\
    \label{eq:b2_def} b_2(\bbV, \epsilon) & := \frac{\mu_{\calG}}{\epsilon \mu_M}(1 - \mu_{\calR}).
\end{align}
We denote $\brho^{(0)} - \brho^*$ and $\by^{(0)} - \by^*$ by $\be^{(0)}_{\brho}$ and $\be^{(0)}_{\by}$, respectively.
The $\ell_2$-norm of the initial distance metric for the encoded representations,  $\|\be^{(0)}_{\by}\|$, relates to $\calG$ and $\calH$ as $\|\be^{(0)}_{\by}\| = \|\calG(\brho^{(0)}) - \calH^{-1}(\brho^*))\|$ where $\by^* = \calH^{-1}(\brho^*)$.
$\be^{(0)}_{\by}$ is therefore dependent on the properties of both networks, and one of the goals during training is to reduce the value of $\rmdist(\by^{(0)}, \by^*)$, defined in~\eqref{eq:dist_definition}, compared to $\rmdist(\brho^{(0)}, \brho^*)$.

Suppose $\chi\in\mathbbm{R}^+$ is the smallest value for which $\|\by^{(0)} - \by^*\| \leq \chi\|\be^{(0)}_{\brho}\|$ for all $\brho^*\in\bbT$, i.e., there is at least one $\brho^*$ for which this upper bound holds with equality.
Then, by using the upper bound expression from~\eqref{eq:dist_rho0_rhoc4}, the Lipschitz constant of the decoder and the assumption that $\calH(\bzero) = \bzero$, we can write
\begin{align}
    \label{eq:e0y_norm_ub} \|\be^{(0)}_{\by}\| \leq \chi\epsilon\mu_{\calH}\|\by^*\|.
\end{align}
Moreover, since $\|\by^{(0)} - \by^*\| \geq |\|\by^{(0)}\| - \|\by^*\||$, we can write
\begin{align}
\label{eq:calG_rho0_calH_inv} \|\by^*\| - \chi\|\be^{(0)}_{\brho}\| & \leq \|\by^{(0)}\| \leq \|\by^*\| + \chi\|\be^{(0)}_{\brho}\|.
\end{align}
Due to the bijective property of $\calH$ over $\bbT$, we have $\|\by^*\| \geq \frac{1}{\mu_{\calH}}\|\brho^*\|$.
Therefore, from~\eqref{eq:calG_rho0_calH_inv},
\begin{align}
\label{eq:calG_rho0_calH_inv2} \frac{1}{\mu_{\calH}} - \chi\epsilon & \leq \frac{\|\by^{(0)}\|}{\|\brho^*\|} \leq \frac{\| \by^* \|}{\|\brho^*\|} + \chi\epsilon.
\end{align}
Now, from~\eqref{eq:dist_rho0_rhoc4}, we have $(1 - \epsilon)\|\brho^*\| \leq \|\brho^{(0)}\| \leq (1 + \epsilon)\|\brho^*\|$ for $0 \leq \epsilon < 1$, and $0 \leq \|\brho^{(0)}\| \leq (1 + \epsilon)\|\brho^*\|$ for $\epsilon \geq 1$.
Therefore, since $\|\by^{(0)}\| \leq \mu_{\calG}(1 + \epsilon)\|\brho^*\|$, by using the lower bound expression from~\eqref{eq:calG_rho0_calH_inv2}, we can write
\begin{align}
    \left(\frac{1}{\mu_{\calH}} - \chi\epsilon\right) & \leq \mu_{\calG}(1 + \epsilon),
\end{align}
which upon rearrangement results in $\chi \geq b_1(\mu_{\calG}, \mu_{\calH}, \epsilon)$.

Moreover, in the upper bound expression in~\eqref{eq:calG_rho0_calH_inv2}, we can use the fact that $\by^* = R \circ \calG(\brho^{(0)})$ if the gradient descent updates in the RNN part of the inversion network converge to the correct encoded image $\by^*$.
Since $\calH(\bzero) = \bzero$, from~\eqref{eq:nablaJ}, we observe that the output of a gradient descent step in the RNN will be zero for a zero input vector, and this leads to $\calR(\bzero) = \bzero$.
Therefore, by expressing $\| \calR(\by^{(0)})\|$ as $\| \calR(\by^{(0)}) - \calR(\bzero)\|$, we can upper bound it by $\mu_{\calR}\mu_{\calG}\|\brho^{(0)}\|$.
Using this expression and the upper bound relation from~\eqref{eq:calG_rho0_calH_inv2}, we observe that the following inequality relation is true:
\begin{align}
    \label{eq:calG_rho0_calH_inv3} \|\by^{(0)}\| & \leq \mu_{\calR}\mu_{\calG}\|\brho^{(0)}\| + \chi\epsilon\|\brho^*\|.
\end{align}

Also, since $\|\brho^*\| \leq \mu_M\|\brho^{(0)}\|$, the upper bound from~\eqref{eq:calG_rho0_calH_inv3} can be at most $(1 + \chi\epsilon \mu_{\calH})\mu_{\calR}\mu_{\calG}\|\brho^{(0)}\|$ for all $\brho^{(0)}\in\bbS$.
From the definition of $\mu_{\calG}$, we know that it is the lowest value for which $\|\by^{(0)}\| \leq \mu_{\calG}\|\brho^{(0)}\|$ for all $\brho^{(0)}\in\bbS$.
This indicates that $\mu_{\calG} \leq (1 + \chi\epsilon \mu_{\calH})\mu_{\calR}\mu_{\calG}$ or $\chi \geq b_2(\bbV, \epsilon)$.
Therefore, we get the inequality relation in~\eqref{eq:k_ub_lb}.
Additionally, for $0 \leq \epsilon < 1$, we know from~\eqref{eq:dist_rho0_rhoc4} that $\|\brho^{(0)}\| \geq (1 - \epsilon)\|\brho^*\|$, and since $\|\brho^*\| \leq \mu_M\|\brho^{(0)}\|$ for all $\brho^{(0)}\in\bbS$, we require that $\frac{1}{\mu_M} \geq (1 - \epsilon)$.

\subsection{Proof of Theorem 1}
\label{app:proof_theorem2}
Let $\bbN_{\epsilon_{\by}}(\by^*)$ denote the $\epsilon_{\by}$-neighborhood of $\by^*$, and it is defined as $\bbN_{\epsilon_{\by}}(\by^*) = \{\by\in\bbY:\rmdist(\by, \by^*)\leq \epsilon_{\by}\|\by^*\|\}$.
For the $l^{th}$ RNN layer, we denote $\by^{(l)} - \by^*$ by $\be^{(l)}_{\by}$ for $l\in[L]$.
Positive real-valued constants $\epsilon_{\by}$ and $\epsilon_{\brho}$ are defined as $\epsilon_{\by} := \chi\mu_{\calH}\epsilon$ and $\epsilon_{\brho} := \mu_M\epsilon_{\by}(1 + \epsilon)$, respectively, where $\chi$ satisfies the inequality relation from~\eqref{eq:k_ub_lb}.
Finally, $\epsilon$ relates to the $\delta$ constant from~\eqref{eq:cond_phaseless2} as shown in (21) in~\cite{yonel2020deterministic}.

\subsubsection{Cost Function Gradient Representation Using~\eqref{eq:cond_phaseless2}}
We begin by presenting a lemma that states two inequality relations for the output vector of the decoder that are used to derive our exact recovery guarantee result under our sufficient condition from~\eqref{eq:cond_phaseless2} and~\eqref{eq:Hcondition}.
\begin{lemma}
\label{lemma:caHy_minus_calHyc_calHy}
Suppose $\calH$ and $\calG$ satisfy the following conditions:
$\calH(\bzero) = \bzero$ and $\calG(\bzero) = \bzero$.
Additionally, $\brho^* = \calH(\by^*)$ for all $\brho^* \in \bbT$, where $\by^* = R \circ \calG(\brho^{(0)})$ and $\brho^{(0)}$ is calculated using spectral initialization scheme from the intensity measurements.
Then, for all $\by\in\bbN_{\epsilon_{\by}}(\by^*)$, we have $\|\calH(\by)\| \leq \mu_{\calH}(1 + \epsilon_{\by})\|\by^*\|$ and
\begin{align}
\label{eq:calH_calH_c} \left(1 - \epsilon_{\brho}\right)\|\calH(\by^*)\| & \leq \|\calH(\by)\| \leq \left(1 + \epsilon_{\brho}\right)\|\calH(\by^*)\|.
\end{align}
\end{lemma}
\begin{proof}
Proof is presented in Subsection~\ref{app:calH_y_calH_yc}.
\end{proof}
This lemma states that for all $\by\in\bbN_{\epsilon_{\by}}(\by^*)$, the corresponding $\calH(\by)$ values in $\bbT$ are within an $\epsilon_{\brho}$-neighbourhood of the unknown ground truth image $\brho^*$.

By using fact that $\be = \calF(\bE_{\brho})$, where $\be$ is defined in Subsection~\ref{subsec:objective}, and the expression of $\frac{1}{M}\calF^H\calF$ from~\eqref{eq:FHF_IRD} in the gradient expression of the cost function $\calK(\by)$ in~\eqref{eq:nablaJ}, we have 
\begin{align}
    \nabla\calK(\by) & = \nabla\calH(\by)\left[\bE_{\brho} + \left(\|\calH(\by)\|^2 - \|\brho^*\|^2\right)\bI + \Delta\left(\bE_{\brho}\right)\right] \nonumber \\
    \label{eq:nablaJ1} & \times \calH(\by).
\end{align}
The relation for the operator $\calQ$ from~\eqref{eq:calR_rho1_rho2} is used in the above expression.
In the following lemma, we present an upper bound on $\|\nabla\calK(\by)\|$ by utilizing the results from Lemma~\ref{lemma:caHy_minus_calHyc_calHy} and the condition on $\calH$ from~\eqref{eq:Hcondition}:
\begin{lemma}
\label{lemma:norm_del_calK}
If the condition in~\eqref{eq:cond_phaseless2} holds for all $\brho^*\in\bbT$, then under the condition on $\calH$ from~\eqref{eq:Hcondition},
\begin{align}
    \label{eq:delJ_norm3} \|\nabla\calK(\by)\| & \leq \mu^4_{\calH}c(\delta, \epsilon_{\by})\|\by^*\|^2\|\be_{\by}\|,
\end{align}
for all $\by\in \bbN_{\epsilon_{\by}}(\by^*)$ where $\brho^* = \calH(\by^*)$ and $c(\delta, \epsilon_{\by}) := (1 + \epsilon_{\by})(2 + \epsilon_{\by})(2 + \omega\delta)$.
\end{lemma}
\begin{proof}
Proof is included in Subsection~\ref{app:norm_derivative}.
\end{proof}

\subsubsection{Regularity Condition and Its Implication}
For our DL based phaseless imaging approach, we show that if a regularity condition similar to the one used in~\cite{yonel2020deterministic} is satisfied for all $\by\in\bbN_{\epsilon_{\by}}(\by^*)$, then with appropriate learning rates, the WF updates remain contractive.
In that case, the iterative updates starting from the initial value $\by^{(0)}$ converge to $\by^*$, from which it can be mapped to the image domain by using the decoding network.
The regularity condition, modified for the encoded representations, is stated below:
\begin{condition}
For all $\by\in \bbN_{\epsilon_{\by}}(\by^*)$ and $\calH(\by^*)\in\bbT$,
\begin{align}
    \label{eq:regularity_cond} \mathrm{Re}\left(\langle\nabla\calK(\by), \be_{\by}\rangle\right) \geq \frac{1}{\alpha}\|\be_{\by}\|^2 + \frac{1}{\beta}\|\nabla\calK(\by)\|^2,
\end{align}
where $\alpha > 0$, $\beta > 0$.
\end{condition}

Using the upper bound on $\|\nabla\calK(\by)\|$ from~\eqref{eq:delJ_norm3}, we observe that this regularity condition is true if
\begin{align}
\label{eq:regularity_cond1} \mathrm{Re}\left(\langle\nabla\calK(\by), \be_{\by}\rangle\right) & \geq q(\alpha, \beta, \by^*)\|\be_{\by}\|^2,
\end{align}
where $q(\alpha, \beta, \by^*) := \frac{1}{\alpha} + \frac{1}{\beta}\mu^8_{\calH}c^2(\delta, \epsilon_{\by})\| \by^*\|^4$.
Since $\nabla\calK(\by)$ equals to $0$ at $\by = \by^*$, the relation in~\eqref{eq:regularity_cond1} can be written as
\begin{align}
    \mathrm{Re}\left(\langle\nabla\calK(\by) - \nabla\calK(\by)_{\by = \by^*}, \be_{\by}\rangle\right) \geq q(\alpha, \beta, \by^*)\|\be_{\by}\|^2
\end{align}
Therefore,~\eqref{eq:regularity_cond1} implies that $\calK(\by)$ is strongly convex in $\bbN_{\epsilon_{\by}}(\by^*)$, and using the definition of strong convexity, it is equivalent to
\begin{align}
    \label{eq:regularity_cond2} \calK(\by) \geq
    0.5q(\alpha, \beta, \by^*)\|\be_{\by}\|^2.
\end{align}

Now, from the definition of Lipschitz constant, we know that $\mu_{\calH} \geq \frac{\|\calH(\by_1) - \calH(\by_2)\|}{\|\by_1 - \by_2\|}$ for all $\by_1, \by_2\in\bbY$.
Similarly, we define another constant $\tilde{\mu}_{\calH}$ as $\tilde{\mu}_{\calH} = \min_{\by_1, \by_2\in\bbY}\frac{\|\calH(\by_1) - \calH(\by_2)\|}{\|\by_1 - \by_2\|}$, which implies that $\tilde{\mu}_{\calH} \leq \frac{\|\calH(\by_1) - \calH(\by_2)\|}{\|\by_1 - \by_2\|}$ for any $\by_1, \by_2\in\bbY$.
Since $\brho^*$ and $\brho$ are reproducible by the decoding network from $\by^*$ and $\by$, respectively, and $\calH(\mathbf{0}) = \mathbf{0}$, we can therefore write
\begin{align}
\label{eq:rho_star_l2_lb} \|\brho^*\|^2 & \geq \tilde{\mu}^2_{\calH}\|\by^*\|^2, \\
\label{eq:e_rho_l2_lb} \|\be_{\brho}\|^2 & \geq \tilde{\mu}^2_{\calH}\|\be_{\by}\|^2.
\end{align}
Then, for the condition in~\eqref{eq:cond_phaseless2}, we can show that $\calK(\by)$ is lower bounded as follows:
\begin{align}
    \label{eq:regularity_cond3} \calK(\by) & \geq
    0.5\tilde{h}(\delta, \epsilon_{\by})\|\brho^*\|^2\|\be_{\brho}\|^2,
\end{align}
the proof of which is presented in Subsection~\ref{app:calK_lb}.
$\tilde{h}(\delta, \epsilon_{\by})$ in~\eqref{eq:regularity_cond3} is defined as
\begin{align}
\tilde{h}(\delta, \epsilon_{\by}) := \left(1 - \delta_1\right)(1 - \epsilon_{\brho})(2 - \epsilon_{\brho}),
\end{align}
and
\begin{align}
    \label{eq:delta_1_def} \delta_1 & := \frac{\sqrt{2}\hat{\delta}(2 + \epsilon_{\brho})(2 + \epsilon_{\by})}{\tilde{\mu}^2_{\calH}(1 - \epsilon_{\brho})(2 - \epsilon_{\brho})}.
\end{align}
Now, for the regularity condition to be redundant, the values of $\alpha$ and $\beta$ in~\eqref{eq:regularity_cond2} should be such that it holds for all $\by^*\in\bbY$.
This can be ensured by setting $\alpha$ and $\beta$ values such that the lower bound in~\eqref{eq:regularity_cond2} is smaller than the lower bound in~\eqref{eq:regularity_cond3}, i.e.,
\begin{align}
	\label{eq:regularity_cond_compare} q(\alpha, \beta, \by^*)\|\be_{\by}\|^2 & \leq \tilde{h}(\delta, \epsilon_{\by})\|\brho^*\|^2\|\be_{\brho}\|^2.
\end{align}
Finally, by using the two lower bound expressions from~\eqref{eq:rho_star_l2_lb} and~\eqref{eq:e_rho_l2_lb}, we can infer from~\eqref{eq:regularity_cond_compare} that the regularity condition is satisfied if $\alpha$ and $\beta$ values are such that $q(\alpha, \beta, \by^*)\|\be_{\by}\|^2$ is upper bounded by $\tilde{h}(\delta, \epsilon_{\by})\tilde{\mu}^4_{\calH}\|\by^*\|^2\|\be_{\by}\|^2$ or equivalently, if
\begin{align}
    \label{eq:regularity_cond_compare2} \frac{1}{\alpha\|\by^*\|^2} + \frac{1}{\beta}\mu^8_{\calH}c^2(\delta, \epsilon_{\by})\| \by^*\|^2 & \leq h(\delta, \epsilon_{\by}),
\end{align}
where $h(\delta, \epsilon_{\by}) := \tilde{\mu}^4_{\calH}\tilde{h}(\delta, \epsilon_{\by})$.

Next, we denote ${\gamma_{l}}/{\|\by^{(0)}\|^2}$ by $\gamma'_l$ for $l\in[L]$.
Since $\|\be^{(l)}_{\by}\|^2$ can be expanded as $\|\left(\by^{(l - 1)} - \gamma'_l\nabla\calK_{\by = \by^{(l - 1)}}\right) - \by^*\|^2$, for $\alpha$ and $\beta$ values that satisfies~\eqref{eq:regularity_cond_compare2}, we can write
\begin{align}
    & \|\be^{(l)}_{\by}\|^2
    =^{(a)} \|\be^{(l - 1)}_{\by}\|^2 - 2\gamma'_l Re(\langle\nabla\calK_{\by = \by^{(l - 1)}}, \be^{(l - 1)}_{\by}\rangle) \nonumber \\
    \label{eq:Cl_2_Vert_yl_yc_Vert} & + \gamma'^2_l\|\nabla\calK_{\by = \by^{(l - 1)}}\|^2, \\
    & \leq^{(b)} \left[1 - \frac{2\gamma'_l}{\alpha} + \left(\gamma'^2_l - \frac{2\gamma'_l}{\beta}\right)\mu^8_{\calH}c^2(\delta, \epsilon_{\by})\| \by^*\|^4\right] \nonumber \\
    \label{eq:convergence_rate} & \times \|\be^{(l - 1)}_{\by}\|^2 \leq^{(c)} \left[1 - \frac{2\gamma'_l}{\alpha}\right]\|\be^{(l - 1)}_{\by}\|^2.
\end{align}
The equality in (a) arises from simple expansion of the squared $\ell_2$-norm expression.
(b) results from using the lower bound on $\rmRe(\langle\nabla\calK(\by)_{\by = \by^{(l)}}, \be^{(l - 1)}_{\by}\rangle)$ from the regularity condition in \eqref{eq:regularity_cond1}, and the upper bound on $\|\nabla\calK(\by)_{\by = \by^{(l - 1)}}\|$ from Lemma
\ref{lemma:norm_del_calK}.
The inequality in (c) results from the assumption that for $l\in[L]$, $\gamma'_l \leq \frac{2}{\beta}$ and $\gamma'_l \geq 0$.
From~\eqref{eq:convergence_rate}, by combining the expressions for
all $l$
up to $j \in [L]$, we can write
\begin{align}
    \|\be^{(j)}_{\by}\|^2 \leq \left[\prod_{l = 1}^{j}\left(1 - \frac{2\gamma'_{l}}{\alpha}\right)\right]\|\be^{(0)}_{\by}\|^2.
\end{align}
By using the inequality from~\eqref{eq:dist_y}, this upper bound expression can be modified to get the relation in~\eqref{eq:convergence_rate3}.

\subsection{Proof of~\eqref{eq:MGMH_ub_lb} and~\eqref{eq:muR_ub}}
\label{app:proof_42_44}
As stated earlier, it is desirable to have a value of $\chi\mu_{\calH}$ between $0$ to $1$.
Suppose the value of $\chi$ learned through training is denoted by $\tau\in\mathbbm{R}^+$.
As shown in Appendix~\ref{app:initialization}, $\tau$ will be larger than or equal to $\max\left(b_1(\mu_{\calG}, \mu_{\calH}, \epsilon), b_2(\bbV, \epsilon)\right)$.
This implies that for a particular $\epsilon$, the Lipschitz constants of $\calG$, $\calH$ and $\calR$ should be such that the maximum of $b_1(\mu_{\calG}, \mu_{\calH}, \epsilon)$ and $b_2(\bbV, \epsilon)$ is positive and as close to $0$ as possible.
On the other hand, since our goal is to have $\epsilon_{\by}$ to be less than $\epsilon$, it is desirable that $\tau\mu_{\calH}$ is upper bounded by a real-valued constant $\xi_{\by}$ as $\tau\mu_{\calH} \leq \xi_{\by} \leq 1$.
Additionally, for $\calH(\by)$ to be within an $\epsilon_{\brho}$-neighbourhood of $\brho^*$ for all $\by\in\bbN_{\epsilon_{\by}}(\by^*)$, where $\epsilon_{\brho} < \epsilon$, we have the following additional condition:
\begin{align}
\label{eq:k_MH_MRNN_MG_1_e} \tau\mu_{\calH}\mu_M(1 + \epsilon) & \leq 1.
\end{align}
Suppose $\tau\mu_{\calH}\mu_M(1 + \epsilon)$ is upper bounded by $\xi_{\brho}$ where $0 < \xi_{\brho} \leq 1$.

For conveniently imposing bounds on the Lipschitz constants of $\calG$, $\calH$ and $\calR$ that achieve desirable encoded image initialization accuracy as well as satisfy~\eqref{eq:k_MH_MRNN_MG_1_e}, we aim to set their values so that $b_1(\mu_{\calG}, \mu_{\calH}, \epsilon) \geq b_2(\bbV, \epsilon)$ and $0 \leq b_1(\mu_{\calG}, \mu_{\calH}, \epsilon) \leq \tau$.
Now, by using the expressions of these two functions from~\eqref{eq:b1_def} and~\eqref{eq:b2_def}, we observe that for $b_1(\mu_{\calG}, \mu_{\calH}, \epsilon)$ to be greater than or equal to $b_2(\bbV, \epsilon)$, we require that
\begin{align}
    \label{eq:muG_muH_ub1} \mu_{\calG}\mu_{\calH} \leq \frac{1}{(1 + \epsilon)}\left(2 - \frac{1}{\mu_{\calR}}\right),
\end{align}
while in order to have $b_1(\mu_{\calG}, \mu_{\calH}, \epsilon) \geq 0$, we need that
\begin{align}
    \label{eq:muG_muH_ub2} \mu_{\calG}\mu_{\calH} \leq {1}/{(1 + \epsilon)}.
\end{align}
Moreover, while deriving the relation in~\eqref{eq:k_ub_lb} in Appendix~\ref{app:initialization}, we observed the necessary condition that $\frac{1}{\mu_M} \geq (1 - \epsilon)$ which always holds if $\epsilon \geq 1$, but becomes significant when $0 \leq \epsilon < 1$ as it requires that
\begin{align}
    \label{eq:muG_muH_ub3} \mu_{\calG}\mu_{\calH} \leq {1}/{[(1 - \epsilon)\mu_{\calR}]}.
\end{align}
The upper bounds in~\eqref{eq:muG_muH_ub1} and~\eqref{eq:muG_muH_ub3} depend on the Lipschitz constant of $\calR$.
Suppose we set $\mu_{\calR}$ so that it is upper bounded by $1$.
In that case, we observe that the upper bound expression in~\eqref{eq:muG_muH_ub1} is less than $\frac{1}{(1 + \epsilon)}$ while the upper bound in~\eqref{eq:muG_muH_ub3} is greater than $\frac{1}{(1 - \epsilon)}$.
This indicates that~\eqref{eq:muG_muH_ub1} is a tighter upper bound requirement than~\eqref{eq:muG_muH_ub2}.
For the case when $0\leq\epsilon<1$, $\frac{1}{(1 + \epsilon)}$ is a tighter upper bound on $\mu_{\calG}\mu_{\calH}$ compared to $\frac{1}{(1 - \epsilon)}$, and therefore, the later becomes redundant.
Therefore, as long as $\mu_{\calG}\mu_{\calH}$ is less than or equal to $\frac{1}{(1 + \epsilon)}\left(2 - \frac{1}{\mu_{\calR}}\right)$, the desired inequality relations in~\eqref{eq:muG_muH_ub1},~\eqref{eq:muG_muH_ub2} and~\eqref{eq:muG_muH_ub3} are satisfied.

Additionally, since another one of our objective is to achieve $\tau\mu_{\calH}\mu_M(1 + \epsilon) \leq \xi_{\brho}$, we first consider the maximum possible value of the term on the left hand side of this inequality and then set it to be less than or equal to $\xi_{\brho}$.
We are upper bounding $\mu_{\calR}$ by $1$ and $\tau\mu_{\calH}$ by $\xi_{\by}$.
Therefore, the maximum possible value that $\tau\mu_{\calH}\mu_M(1 + \epsilon)$ can take is $\xi_{\by}\mu_{\calG}\mu_{\calH}(1 + \epsilon)$.
Upon setting this expression to be less than $\xi_{\brho}$ and rearrangement, we get $\mu_{\calG}\mu_{\calH} \leq {\xi_{\brho}}/{[\xi_{\by}(1 + \epsilon)]}$.
Finally, by setting $b_1(\mu_{\calG}, \mu_{\calH}, \epsilon)$ to be less than $\tau$, we get a lower bound on $\mu_{\calG}\mu_{\calH}$ as $\mu_{\calG}\mu_{\calH} \geq {(1 - \tau\epsilon\mu_{\calH})}/{(1 + \epsilon)}$.

\subsection{Proof of Lemma~\ref{lemma:caHy_minus_calHyc_calHy}}
\label{app:calH_y_calH_yc}
For all $\by\in\bbN_{\epsilon_{\by}(\by^*)}$, we have $\|\be_{\by}\|\leq\epsilon_{\by}\|\by^*\|$.
Since by using the Lipschitz constant of the decoding network, $\frac{1}{\mu_{\calH}}\|\be_{\brho}\|\leq\|\be_{\by}\|$, therefore 
$\|\be_{\brho}\|$ is upper bounded by $\epsilon_{\by} \mu_{\calH}\|\by^*\|$.
On the other hand, by using the triangular inequality, we can write $|\|\calH(\by)\| - \|\calH(\by^*)\|| \leq \|\be_{\brho}\|$.
Combining this relation with the upper bound on $\|\be_{\brho}\|$, we can write $\|\calH(\by)\| \leq \|\calH(\by^*)\| + \epsilon_{\by} \mu_{\calH}\|\by^*\|$ and $\|\calH(\by)\| \geq \|\calH(\by^*)\| - \epsilon_{\by}\mu_{\calH}\|\by^*\|$.
We modify these upper and lower bound expressions further by using the Lipschitz constant property of the encoder and the RNN.
Since $\by^* = \calR \circ \calG(\brho^{(0)})$, we use the inequality $\|\by^*\| \leq \mu_{\calR}\mu_{\calG}(1 + \epsilon)\|\brho^*\|$ and get $\|\calH(\by)\| \leq \left(1 + \epsilon_{\brho}\right)\|\calH(\by^*)\|$ and $\|\calH(\by)\| \geq \left(1 - \epsilon_{\brho}\right)\|\calH(\by^*)\|$.
Additionally, by using the Lipschitz constant of the decoding network, we can modify this upper bound expression as $\|\calH(\by)\| \leq \mu_{\calH}(1 + \epsilon_{\brho})\|\by^*\|$.

\vspace{-0.1in}
\subsection{Proof of Lemma~\ref{lemma:norm_del_calK}}
\label{app:norm_derivative}
By taking the $\ell_2$-norm of both sides of~\eqref{eq:nablaJ1}, we can upper bound $\|\nabla\calK\|$ as
\begin{align}
    \label{eq:delJ_norm} \|\nabla\calK\| & \leq \|\nabla\calH(\by)\|\left(q_1(\by) + q_2(\by) + q_3(\by)\right),
\end{align}
where
\begin{align}
    q_1(\by) & := \|\bE_{\brho}\calH(\by)\|, \\
    q_2(\by) & := \|\Delta\left(\bE_{\brho}\right)\calH(\by)\|, \\
    q_3(\by) & := \|\left(\|\calH(\by)\|^2 - \|\brho^*\|^2\right)\calH(\by)\|.
\end{align}
Now, by expressing $\bE_{\brho}$ as $\be_{\brho}\calH(\by)^H + \brho^*\be^H_{\brho}$, it is easy to verify that $q_1(\by)$ is upper bounded by $\|\calH(\by)\|\|\be_{\brho}\|(\|\calH(\by)\| + \|\brho^*\|)$.
Then, by using the upper bound on $\|\calH(\by)\|$ from Lemma~\ref{lemma:caHy_minus_calHyc_calHy} in this expression, we get
\begin{align}
    \label{eq:part1_2} q_1(\by) & \leq \mu_{\calH}(2 + \epsilon_{\by})\|\calH(\by)\|\|\be_{\brho}\|\|\by^*\|.
\end{align}

Next, we upper bound the expression of $q_2(\by)$ by $\|\Delta\left(\bE_{\brho}\right)\|\|\calH(\by)\|$.
Under our sufficient condition from~\eqref{eq:cond_phaseless2} and~\eqref{eq:Hcondition}, $\|\Delta\left(\bE_{\brho}\right)\|$ is upper bounded by $\hat{\delta}\|\tilde{\bY} - \tilde{\bY}^*\|_F$ where $\hat{\delta} = \omega\mu^2_{\calH}\delta$.
We can then determine an upper bound on $\|\tilde{\bY} - \tilde{\bY}^*\|_F$ as a constant multiplier of $\|\be_{\by}\|\|\by^*\|$ as $\|\tilde{\bY} - \tilde{\bY}^*\|_F \leq (2 + \epsilon_{\by})\|\be_{\by}\|\|\by^*\|$,
from which we conclude that
\begin{align}
    q_2(\by) & \leq \hat{\delta}(2 + \epsilon_{\by})\|\calH(\by)\|\|\be_{\by}\|\|\by^*\|.
\end{align}
Finally, by using the inequalities $\|\by\| \leq (1 + \epsilon_{\by})\|\by^*\|$ and $|\|\calH(\by)\| - \|\brho^*\|| \leq \|\be_{\brho}\|$, we can upper bound $q_3(\by)$ as follows:
\begin{align}
    q_3(\by) & \leq \mu_{\calH}(2 + \epsilon_{\by})\|\calH(\by)\|\|\be_{\brho}\|\|\by^*\|.
\end{align}

After replacing the upper bound expressions for $q_1(\by)$, $q_2(\by)$ and $q_3(\by)$ in the right hand side of~\eqref{eq:delJ_norm} and by using the Lipschitz constant of $\calH$, we get
\begin{align}
    \label{eq:delJ_norm2} & \|\nabla\calK(\by)\|
    \leq \mu^3_{\calH}c(\delta, \epsilon_{\by})\|\nabla\calH(\by)\|\|\by^*\|^2\|\be_{\by}\|.
\end{align}
From the definition of Lipschitz constant, $\mu_{\calH}\geq\|\nabla\calH(\by)\|$ for all $\by\in\bbY$.
Therefore, from~\eqref{eq:delJ_norm2}, we get the following upper bound relation:
\begin{align}
    \|\nabla\calK(\by)\| & \leq \mu^4_{\calH}c(\delta, \epsilon_{\by})\|\by^*\|^2\|\be_{\by}\|.
\end{align}

\subsection{Proof of the Inequality Relation in~\eqref{eq:regularity_cond3}}
\label{app:calK_lb}
From the definition of $\calK(\by)$ in~\eqref{eq:calK}, $\calK(\by)$ can be expressed as a function of $\bE_{\brho}$ by
$\frac{1}{2M}\|\calF\left(\bE_{\brho}\right)\|^2$.
Since we can write $\frac{1}{M}\|\calF\left(\bE_{\brho}\right)\|^2$ as $\langle\frac{1}{M}\calF^H\calF(\bE_{\brho}), \bE_{\brho}\rangle$, by using the results from Lemma~\ref{lemma:lemma_III_1}, we can verify that
\begin{align}
    \frac{1}{M}\|\calF\left(\bE_{\brho}\right)\|^2 \geq \|\bE_{\brho}\|^2_F + \langle\Delta(\bE_{\brho}), \bE_{\brho}\rangle_F.
\end{align}
Now, since
\begin{align}
    |\langle\Delta(\bE_{\brho}), \bE_{\brho}\rangle_F| & \leq \sqrt{2}\|\bE_{\brho}\|_F\|\Delta(\bE_{\brho})\|,
\end{align}
and under our sufficient condition from~\eqref{eq:cond_phaseless2} and~\eqref{eq:Hcondition}, we know that $\|\Delta(\bE_{\brho})\| \leq \hat{\delta}\|\tilde{\bY} - \tilde{\bY}^*\|_F$, then we observe that $|\langle\Delta(\bE_{\brho}), \bE_{\brho}\rangle_F|$ is upper bounded by $\sqrt{2}\hat{\delta}\|\bE_{\brho}\|_F\|\tilde{\bY} - \tilde{\bY}^*\|_F$.
Using similar steps as the proof of Lemma III.3 in~\cite{yonel2020deterministic}, it is easy to verify that
\begin{align}
    \label{eq:Erho_F_ub} \|\bE_{\brho}\|_F & \leq (2 + \epsilon_{\brho})\|\be_{\brho}\|\|\brho^*\|, \\
    \label{eq:Erho_F_lb} \|\bE_{\brho}\|_F & \geq \sqrt{(1 - \epsilon_{\brho})(2 - \epsilon_{\brho})}\|\brho^*\|\|\be_{\brho}\|, \\
    \label{eq:Ey_F_ub} \|\tilde{\bY} - \tilde{\bY}^*\|_F & \leq (2 + \epsilon_{\by})\|\be_{\by}\|\|\by^*\|.
\end{align}

Using the two upper bounds from~\eqref{eq:Erho_F_ub} and~\eqref{eq:Ey_F_ub}, the upper bound expression on $|\langle\Delta(\bE_{\brho}), \bE_{\brho}\rangle_F|$ can be modified to $\delta_1\|\bE_{\brho}\|^2_F$.
Finally, by using the lower bound expression from~\eqref{eq:Erho_F_lb}, we observe that $\frac{1}{M}\|\calF\left(\bE_{\brho}\right)\|^2$ is lower bounded by $(1 - \delta_1)(1 - \epsilon_{\brho})(2 - \epsilon_{\brho})\|\brho^*\|^2\|\be_{\brho}\|^2$.

\bibliographystyle{IEEEtran}
\bibliography{references}

\end{document}